\PassOptionsToPackage{unicode,hidelinks}{hyperref}
\PassOptionsToPackage{hyphens}{url}
\PassOptionsToPackage{dvipsnames,svgnames,x11names}{xcolor}
\documentclass[
  12pt]{article}

\usepackage{amsmath,amssymb}
\usepackage{iftex}
\ifPDFTeX
  \usepackage[T1]{fontenc}
  \usepackage[utf8]{inputenc}
  \usepackage{textcomp} % provide euro and other symbols
\else % if luatex or xetex
  \usepackage{unicode-math}
  \defaultfontfeatures{Scale=MatchLowercase}
  \defaultfontfeatures[\rmfamily]{Ligatures=TeX,Scale=1}
\fi
\usepackage{newtxtext}
\ifPDFTeX\else  
\fi
\IfFileExists{upquote.sty}{\usepackage{upquote}}{}
\IfFileExists{microtype.sty}{% use microtype if available
  \usepackage[]{microtype}
  \UseMicrotypeSet[protrusion]{basicmath} % disable protrusion for tt fonts
}{}
\makeatletter
\@ifundefined{KOMAClassName}{% if non-KOMA class
  \IfFileExists{parskip.sty}{%
    \usepackage{parskip}
  }{% else
    \setlength{\parindent}{0pt}
    \setlength{\parskip}{6pt plus 2pt minus 1pt}}
}{% if KOMA class
  \KOMAoptions{parskip=half}}
\makeatother
\usepackage{xcolor}
\makeatletter
\ifx\paragraph\undefined\else
  \let\oldparagraph\paragraph
  \renewcommand{\paragraph}{
    \@ifstar
      \xxxParagraphStar
      \xxxParagraphNoStar
  }
  \newcommand{\xxxParagraphStar}[1]{\oldparagraph*{#1}\mbox{}}
  \newcommand{\xxxParagraphNoStar}[1]{\oldparagraph{#1}\mbox{}}
\fi
\ifx\subparagraph\undefined\else
  \let\oldsubparagraph\subparagraph
  \renewcommand{\subparagraph}{
    \@ifstar
      \xxxSubParagraphStar
      \xxxSubParagraphNoStar
  }
  \newcommand{\xxxSubParagraphStar}[1]{\oldsubparagraph*{#1}\mbox{}}
  \newcommand{\xxxSubParagraphNoStar}[1]{\oldsubparagraph{#1}\mbox{}}
\fi
\makeatother

\usepackage{longtable,booktabs,array}
\usepackage{calc} % for calculating minipage widths
\usepackage{etoolbox}
\makeatletter
\patchcmd\longtable{\par}{\if@noskipsec\mbox{}\fi\par}{}{}
\makeatother
\IfFileExists{footnotehyper.sty}{\usepackage{footnotehyper}}{\usepackage{footnote}}
\makesavenoteenv{longtable}
\usepackage{graphicx}
\makeatletter
\def\maxwidth{\ifdim\Gin@nat@width>\linewidth\linewidth\else\Gin@nat@width\fi}
\def\maxheight{\ifdim\Gin@nat@height>\textheight\textheight\else\Gin@nat@height\fi}
\makeatother
\setkeys{Gin}{width=\maxwidth,height=\maxheight,keepaspectratio}
\makeatletter
\def\fps@figure{htbp}
\makeatother

\makeatletter
\@ifpackageloaded{caption}{}{\usepackage{caption}}
\AtBeginDocument{%
\ifdefined\contentsname
  \renewcommand*\contentsname{Table of contents}
\else
  \newcommand\contentsname{Table of contents}
\fi
\ifdefined\listfigurename
  \renewcommand*\listfigurename{List of Figures}
\else
  \newcommand\listfigurename{List of Figures}
\fi
\ifdefined\listtablename
  \renewcommand*\listtablename{List of Tables}
\else
  \newcommand\listtablename{List of Tables}
\fi
\ifdefined\figurename
  \renewcommand*\figurename{Figure}
\else
  \newcommand\figurename{Figure}
\fi
\ifdefined\tablename
  \renewcommand*\tablename{Table}
\else
  \newcommand\tablename{Table}
\fi
}
\@ifpackageloaded{float}{}{\usepackage{float}}
\floatstyle{ruled}
\@ifundefined{c@chapter}{\newfloat{codelisting}{h}{lop}}{\newfloat{codelisting}{h}{lop}[chapter]}
\floatname{codelisting}{Listing}

\makeatother
\makeatletter
\@ifpackageloaded{caption}{}{\usepackage{caption}}
\@ifpackageloaded{subcaption}{}{\usepackage{subcaption}}
\makeatother

\ifLuaTeX
  \usepackage{selnolig}  % disable illegal ligatures
\fi
\usepackage[]{natbib}
\usepackage{bookmark}

\IfFileExists{xurl.sty}{\usepackage{xurl}}{} % add URL line breaks if available
\hypersetup{
  pdftitle={Title},
  pdfauthor={Author 1; Author 2},
  pdfkeywords={3 to 6 keywords, that do not appear in the title},
  colorlinks=true,
  linkcolor={blue},
  filecolor={Maroon},
  citecolor={Blue},
  urlcolor={Blue},
  pdfcreator={LaTeX via pandoc}}

\usepackage{bm}
\usepackage{mathtools}
\mathtoolsset{showonlyrefs=true}
\usepackage{amsthm}
\newtheorem{assumption}{Assumption}

\newtheorem{proposition}{Proposition}
\newtheorem{corollary}{Corollary}
\newtheorem{supproposition}{Supplementary Proposition}

\theoremstyle{remark}
\newtheorem{remark}{Remark}

\newcommand{\1}{\mathbb{I}}
\newcommand{\E}{\mathbb{E}}
\newcommand{\Var}{\mathrm{Var}}
\newcommand{\Cov}{\mathrm{Cov}}

\newcommand{\R}{\mathbb{R}}
\newcommand{\iid}{\stackrel{\text{iid}}{\sim}}

\newcommand{\calD}{\mathcal{D}}
\newcommand{\calN}{\mathcal{N}}
\newcommand{\calA}{\mathcal{A}}
\newcommand{\gauss}{\mathcal{N}}

\newcommand{\bQ}{\bm{Q}}
\newcommand{\bW}{\bm{W}}
\newcommand{\bD}{\bm{D}}
\newcommand{\bI}{\bm{I}}

\newcommand{\bu}{\bm{u}}
\newcommand{\beps}{\bm{\varepsilon}}
\newcommand{\bOne}{\bm{1}}

\usepackage[hidelinks]{hyperref}

\newcommand{\anon}{1}

\begin{document}

\def\spacingset#1{\renewcommand{\baselinestretch}%
{#1}\small\normalsize} \spacingset{1}

%%%%%%%%%%%%%%%%%%%%%%%%%%%%%%%%%%%%%%%%%%%%%%%%%%%%%%%%%%%%%%%%%%%%%%%%%%%%%%

\if1\anon
{
  \title{\bf A Bayesian Longitudinal Spatial Normative Model For Individualized Brain Deviation Mapping}
  \author{Joshua T. Korley\thanks{
    This research received no specific grant from any funding agency in the public, commercial, or not-for-profit sectors. The author thanks colleagues and peers who provided helpful feedback and discussion during the development of this work. The author also acknowledges the OASIS-3 investigators and participants for making the neuroimaging and clinical data publicly available}\hspace{.2cm}\\
    Department of Epidemiology \& Biostatistics, University of South Carolina\\
    % and \\
    % Author 2 \\
    % Department of ZZZ, University of WWW
    \texttt{jkorley@email.sc.edu}
    }
    \date{}
  \maketitle
} \fi

\if0\anon
{
  \bigskip
  \bigskip
  \bigskip
  \begin{center}
    {\LARGE\bf A Bayesian Longitudinal Spatial Normative Model}
\end{center}
  \medskip
} \fi

\bigskip
\begin{abstract}
Normative modeling enables individualized characterization of structural brain deviations by evaluating subjects against a reference population rather than a group average. Most existing implementations treat brain regions independently and remain cross-sectional, despite the availability of repeated neuroimaging measurements and the well-documented spatial organization of neuroanatomical variation. We propose a Bayesian longitudinal spatial normative model that jointly captures within-subject temporal dependence and spatially structured subject-specific deviations within a unified hierarchical framework. The individualized deviation map is treated as a latent spatial process with an explicit posterior distribution, yielding a principled Bayes estimator under squared error loss rather than an ad hoc residual summary. Across six simulation scenarios encompassing varying spatial dependence, nonlinear trajectories, irregular visit schedules, and missing follow-up, the proposed model consistently reduced deviation-map reconstruction error relative to cross-sectional and longitudinal non-spatial benchmarks while maintaining stable calibration. In an application to OASIS-3 structural MRI data, the model reduced RMSE by 54\% and 45\% relative to the two benchmarks. Regional deviation burden was concentrated in the temporal pole, entorhinal cortex, inferior temporal cortex, posterior cingulate, and parahippocampal cortex. Subject-level profiles revealed substantial heterogeneity in regional abnormality patterns, including marked multiregional deviation with preserved global cognitive scores.
\end{abstract}

\noindent%
\vfill

\newpage
\spacingset{1.8} 

\section{Introduction}\label{sec-intro}
Neuroimaging studies have traditionally relied on population-average comparisons between diagnostic groups to characterize structural brain abnormalities. Although these approaches have produced important insights into neurological and psychiatric disease, they are less suited for describing heterogeneity at the individual level. Subjects with similar clinical diagnoses may exhibit markedly different spatial patterns of brain change, while substantial overlap often exists between patient and control populations. These limitations have motivated increasing interest in normative modeling approaches that estimate individualized deviations from a reference population rather than focusing exclusively on group-average effects \citep{Marquand2016,Rutherford2022,Bethlehem2022}.

Normative modeling shifts the focus of inference from group contrasts to individual deviation. The objective is to characterize a reference distribution for brain measures conditional on relevant covariates and to evaluate how each individual departs from this reference. This perspective parallels growth chart methodology, where observations are interpreted relative to expected trajectories. In neuroimaging, this formulation enables the study of heterogeneity within diagnostic categories and supports identification of subjects whose deviation patterns may reflect clinically meaningful or biologically distinct processes \citep{Marquand2016,Rutherford2022,Fraza2021,Verdi2023}.

Methodological developments have expanded this framework. Hierarchical Bayesian approaches improve robustness to site-specific variability and scanner effects \citep{Kia2020}. Flexible regression strategies extend applicability to non-Gaussian response distributions \citep{Fraza2021}. Large-scale studies demonstrate that normative trajectories can be estimated across the lifespan, yielding reference curves for structural brain phenotypes \citep{Bethlehem2022}. These contributions establish normative modeling as a core tool for individualized neuroimaging analysis.

Key structural features of neuroimaging data remain insufficiently addressed. One limitation concerns longitudinal structure. Repeated measurements are increasingly available, yet many implementations remain effectively cross-sectional or rely on sequential procedures that do not fully account for within-subject dependence. This leads to inefficient use of data and unstable deviation estimates when trajectories evolve over time \citep{Buckova2025}.

A second limitation concerns spatial dependence. Brain measurements exhibit structured relationships across anatomically and functionally related regions. Deviations from normative expectations frequently appear in spatially coherent patterns. Treating regions independently neglects this structure, reducing efficiency and distorting the geometry of estimated deviation maps. Spatial Bayesian models address related problems through structured priors that borrow strength across neighboring regions \citep{Huertas2017,Mejia2022}. Spectral approaches extend spatial representation to higher resolution settings \citep{Mansour2025}. These developments are rarely integrated within a normative modeling framework that targets individualized deviation in longitudinal data.

A third limitation concerns the role of deviation itself. In many applications, deviation is defined implicitly as a residual from a fitted model. When deviation is the primary scientific object, this representation is inadequate. In neurodegenerative and neurodevelopmental disorders, spatial configurations of deviation may encode information about disease processes that is not captured by marginal mean effects. A formulation that treats deviation as a structured latent quantity permits direct estimation and principled uncertainty quantification.

These limitations define a single inferential problem: estimation of subject-specific deviation maps that evolve over time and exhibit structured dependence across brain regions. Addressing this problem requires a model that captures longitudinal dynamics, spatial structure, and individual-level variation within a coherent probabilistic framework.

We develop a Bayesian longitudinal spatial normative model for repeated structural neuroimaging measurements. The model combines a covariate-driven normative mean function, a subject-specific random intercept for repeated visits, and a spatially structured latent deviation process across brain regions. This formulation treats individualized deviation as the primary inferential target rather than as a secondary residual summary.

The paper makes three contributions. First, it introduces a region-level Bayesian normative model that jointly captures longitudinal dependence and spatially structured subject-specific deviation. Second, it characterizes the induced covariance structure and the posterior distribution of individualized spatial deviation maps under Gaussian assumptions. Third, it evaluates the model through simulation studies and an OASIS-3 structural MRI application, comparing it with independent cross-sectional and longitudinal non-spatial alternatives.
This work addresses a gap in the literature. Prior research has established the importance of normative modeling for individualized inference \citep{Marquand2016,Rutherford2022,Fraza2021}, the role of hierarchical Bayesian methods for handling multi-site and heterogeneous data \citep{Kia2020}, and the value of spatial and longitudinal modeling when considered separately. In particular, longitudinal neuroimaging studies have emphasized subject-specific trajectories over time, while spatial Bayesian models have demonstrated the importance of borrowing information across anatomically related regions \citep{Mejia2022,Mansour2025}. 

However, these developments have largely progressed in parallel. Existing spatial longitudinal models are primarily formulated for population-level inference, focusing on group differences or disease progression, and do not operate within a normative modeling framework that enables individualized deviation estimation. Conversely, current normative modeling approaches, including recent scalable Bayesian formulations, typically treat brain regions independently or do not explicitly incorporate subject-specific spatial structure.

A region-level framework that jointly models longitudinal trajectories and spatially structured subject-specific deviations, while remaining interpretable and amenable to rigorous evaluation, remains limited. The proposed approach provides such a framework and supports individualized characterization of brain variation over time with quantified uncertainty.

The remainder of the paper is organized as follows. Section~\ref{sec-meth} presents the proposed Bayesian longitudinal spatial normative model, including the hierarchical formulation, posterior inference, and theoretical properties. Section~\ref{sec-simulation} describes the simulation study and comparative evaluation under varying spatial and longitudinal data-generating settings. Section~\ref{sec-oasis} applies the proposed framework to longitudinal structural MRI data from OASIS-3. Section~\ref{sec-discu} concludes with discussion, limitations, and directions for future research. Supplementary material provides additional theoretical derivations, simulation details, diagnostic analyses, and extended real-data results.

\section{Methodology}\label{sec-meth}
\subsection{Model formulation}
We propose a Bayesian longitudinal spatial normative model for region-level structural MRI outcomes. Suppose there are $n$ subjects indexed by $i=1,\dots,n$, with subject $i$ observed at visits $t=1,\dots,T_i$. At each visit, measurements are recorded across $R$ brain regions indexed by $r=1,\dots,R$. Let $Y_{itr}$ denote the structural measurement for subject $i$, visit $t$, and region $r$, and let $\bm{X}_{it}\in\R^p$ denote a visit-level covariate vector. The vector of regional measurements is written as
$\bm{Y}_{it}=(Y_{it1},\dots,Y_{itR})^\top.$

For each subject, visit, and region, the proposed model is
\begin{equation}
Y_{itr}
=
\bm{X}_{it}^{\top}\bm{\beta}_r
+
b_i
+
u_{ir}
+
\varepsilon_{itr},
\label{eq:scalar_model}
\end{equation}
where $\bm{\beta}_r$ is a region-specific regression coefficient vector, $b_i$ is a subject-specific random intercept, $u_{ir}$ is a subject- and region-specific latent spatial deviation, and $\varepsilon_{itr}$ is a residual error term. The stochastic components are
\[
b_i\iid \gauss(0,\sigma_b^2),
\qquad
\bu_i=(u_{i1},\dots,u_{iR})^\top \iid \gauss(\bm{0},\tau_u^2\bQ(\rho)^{-1}),
\qquad
\varepsilon_{itr}\iid \gauss(0,\sigma^2),
\]
with mutual independence across the three latent components.

The spatial precision matrix is defined from a known region-level adjacency matrix. A convenient proper conditional autoregressive specification is
\begin{equation}
\bQ(\rho)=\bD-\rho\bW,
\label{eq:Qrho}
\end{equation}
where $\bW$ is an $R\times R$ symmetric adjacency matrix, $\bD$ is diagonal with entries $D_{rr}=\sum_{s=1}^R W_{rs}$, and $\rho$ lies in the admissible interval that ensures positive definiteness of $\bQ(\rho)$ \citep{Banerjee2014,RueHeld2005}. This construction encourages anatomically related regions to share information while still allowing individualized regional heterogeneity.

In vector form, the model is
\begin{equation}
\bm{Y}_{it}
=
\bm{B}^{\top}\bm{X}_{it}
+
b_i\bOne_R
+
\bu_i
+
\beps_{it},
\label{eq:vector_model}
\end{equation}
where
$\bm{B}=\left(\bm{\beta}_1,\dots,\bm{\beta}_R\right)$
is a \(p \times R\) coefficient matrix whose \(r\)-th column corresponds to the region-specific coefficient vector \(\bm{\beta}_r\), and
$\beps_{it}\sim\gauss(\bm{0},\sigma^2\bI_R)$.
The normative mean function $\bm{X}_{it}^{\top}\bm{\beta}_r$ may be linear or may include spline basis terms for nonlinear age effects. In the simplest specification, $\bm{X}_{it}$ includes an intercept, age, sex, and other relevant covariates. More flexible versions can replace age by a basis expansion without changing the hierarchical structure of the model.

\subsection{Priors, prediction, and deviation scores}

For the main Bayesian formulation, weakly informative Gaussian priors are assigned to the region-specific regression coefficients,
$\bm{\beta}_r \iid \gauss(\bm{0},\sigma_\beta^2\bI_p),$
$ r=1,\dots,R.$
Positive scale parameters may be assigned weakly informative half-Cauchy or exponential priors, depending on the computational implementation:
$\sigma_b>0$, $\sigma>0$, and $\tau_u>0$.
The spatial dependence parameter $\rho$ is assigned a prior supported on the admissible interval of $\bQ(\rho)$.

For a new or held-out observation $(i,t,r)$, the model yields a posterior predictive distribution. Let
$\widehat{\mu}_{itr}=
\E(Y_{itr}^{\mathrm{new}}\mid\calD)$,
$\widehat{v}_{itr}=
\Var(Y_{itr}^{\mathrm{new}}\mid\calD),$
where $\calD$ denotes the observed data. Following the normative probability map construction of \citet{Marquand2016}, we define the standardized deviation score
%\begin{equation}
$Z_{itr}=\frac{Y_{itr}-\widehat{\mu}_{itr}}{\sqrt{\widehat{v}_{itr}}}.$

%\label{eq:zscore_main}
%\end{equation}
Large negative values indicate unexpectedly small structural measurements relative to the normative reference model, while large positive values indicate unexpectedly large measurements. The vector $\bm{Z}_{it}=(Z_{it1},\dots,Z_{itR})^\top$ is the individualized regional deviation map at visit $t$.

A subject-level summary of deviation burden may also be constructed by pooling absolute standardized deviations across visits and regions:
$A_i^{(m)}
=
\frac{1}{m}
\sum_{k=1}^{m}|Z_i|_{(k)},$
where $|Z_i|_{(1)}\ge \cdots \ge |Z_i|_{(L_i)}$ are the ordered absolute deviation scores for subject $i$. This summary is useful for ranking subjects by overall deviation burden, but the primary inferential target remains the full regional deviation map.

\subsection{Theoretical properties}

We record the main properties needed for interpretation of the proposed model. Detailed proofs and additional theoretical results are given in the Supplementary Material.

\begin{assumption}[Reference-model regularity]
\label{ass:regularity}
The data arise from the hierarchical Gaussian model in \eqref{eq:scalar_model}, with $\bQ(\rho)$ positive definite for the true value of $\rho$.
\end{assumption}

\begin{assumption}[Conditional independence]
\label{ass:indep}
Conditional on $(\bm{B},b_i,\bu_i)$, the residual vectors $\beps_{it}$ are independent across visits and subjects.
\end{assumption}

\begin{assumption}[Known adjacency]
\label{ass:adj}
The region-level adjacency matrix $\bW$ is fixed and known. The spatial precision matrix $\bQ(\rho)$ is symmetric positive definite on the admissible parameter space of $\rho$.
\end{assumption}

\begin{proposition}[Marginal mean and covariance]
\label{prop:moments}
Under model \eqref{eq:scalar_model}, the marginal mean of $Y_{itr}$ is
$\E(Y_{itr}\mid\bm{X}_{it})
=
\bm{X}_{it}^{\top}\bm{\beta}_r.$
For any two observations $(i,t,r)$ and $(j,s,\ell)$,
\begin{equation}
\Cov(Y_{itr},Y_{js\ell}\mid\bm{X})
=
\1(i=j)\sigma_b^2
+
\1(i=j)\tau_u^2[\bQ(\rho)^{-1}]_{r\ell}
+
\1(i=j,t=s,r=\ell)\sigma^2.
\label{eq:cov_general}
\end{equation}
\end{proposition}

\begin{proof}
The result follows by subtracting the fixed mean in \eqref{eq:scalar_model} and using independence of $b_i$, $\bu_i$, and $\varepsilon_{itr}$. The random intercept contributes $\sigma_b^2$ to all observations from the same subject, the spatial deviation contributes $\tau_u^2[\bQ(\rho)^{-1}]_{r\ell}$ within subject across regions, and the residual contributes $\sigma^2$ only for the same subject, visit, and region.
\end{proof}

\begin{remark}
Proposition~\ref{prop:moments} separates two sources of dependence. The random intercept induces subject-level dependence across all repeated measurements, while the spatial deviation process induces structured regional dependence within subject. Independent-region models do not capture this decomposition.
\end{remark}

For a fixed subject $i$, define the residualized visit vector, $\widetilde{\bm{Y}}_{it}
=
\bm{Y}_{it}
-
\bm{B}^{\top}\bm{X}_{it}
-
b_i\bOne_R.$
Conditional on $(\bm{B},b_i)$, model \eqref{eq:vector_model} implies
$\widetilde{\bm{Y}}_{it}=
\bu_i+\beps_{it}$, $\beps_{it}\iid \gauss(\bm{0},\sigma^2\bI_R).$

\begin{proposition}[Posterior distribution of subject-specific spatial deviations]
\label{prop:posterior_ui}
Fix subject $i$ and suppose $(\bm{B},b_i,\sigma^2,\tau_u^2,\rho)$ are known. Under \eqref{eq:vector_model}, the conditional posterior distribution of $\bu_i$ given all visits of subject $i$ is
\begin{equation}
\bu_i
\mid
\{\bm{Y}_{it}\}_{t=1}^{T_i},
\bm{B},
b_i,
\sigma^2,
\tau_u^2,
\rho
\sim
\gauss(\bm{m}_i,\bm{S}_i),
\label{eq:post_ui}
\end{equation}
where
\begin{align}
\bm{S}_i
&=
\left(
\tau_u^{-2}\bQ(\rho)
+
T_i\sigma^{-2}\bI_R
\right)^{-1},
\label{eq:Si}\\
\bm{m}_i
&=
\bm{S}_i
\left(
\sigma^{-2}
\sum_{t=1}^{T_i}
\widetilde{\bm{Y}}_{it}
\right).
\label{eq:mi}
\end{align}
\end{proposition}

\begin{proof}
The prior for $\bu_i$ is Gaussian with precision $\tau_u^{-2}\bQ(\rho)$. The likelihood from the $T_i$ residualized visits contributes precision $T_i\sigma^{-2}\bI_R$ and linear term $\sigma^{-2}\sum_{t=1}^{T_i}\widetilde{\bm{Y}}_{it}$. Completing the square yields \eqref{eq:post_ui}--\eqref{eq:mi}. The full algebra is provided in the Supplementary Material.
\end{proof}

\begin{remark}
Proposition~\ref{prop:posterior_ui} shows that the individualized deviation map is a posterior latent quantity, not an ad hoc residual image. The posterior covariance in \eqref{eq:Si} shows how repeated visits and spatial structure jointly influence shrinkage of subject-level deviation estimates.
\end{remark}

\begin{corollary}[Bayes estimator of the individualized deviation map]
\label{cor:bayes_estimator}
Under squared error loss $L(\widehat{\bu}_i,\bu_i)=\|\widehat{\bu}_i-\bu_i\|_2^2,$
the posterior mean
$\widehat{\bu}_i^{\mathrm{Bayes}}
=
\E(\bu_i\mid \{\bm{Y}_{it}\}_{t=1}^{T_i},\bm{B},b_i,\sigma^2,\tau_u^2,\rho)
=
\bm{m}_i$
is the Bayes estimator of the subject-specific spatial deviation map.
\end{corollary}

\begin{remark}
Corollary~\ref{cor:bayes_estimator} gives the Bayesian interpretation of the individualized deviation map used throughout the paper. The posterior mean map is not only a smoothed residual summary; it is the Bayes estimator of the latent subject-specific spatial deviation process under squared error loss. A full decision-theoretic proof is provided in the Supplementary Material.
\end{remark}

\section{Simulation study}\label{sec-simulation}
We evaluated the proposed model through six simulation scenarios representing distinct features of longitudinal neuroimaging data, including varying strengths of spatial dependence, unequal visit schedules, missing follow-up, and nonlinear mean trajectories. The primary goal was to assess whether joint modeling of repeated measurements and spatial dependence improves recovery of subject-specific deviation maps relative to independent cross-sectional and longitudinal non-spatial models.

Each simulated dataset consisted of region-level structural brain measurements observed over repeated visits. The data-generating mechanism matched the hierarchical structure of the proposed model so that the simulation estimand was the subject-specific spatial deviation map. For each Monte Carlo replicate, we generated subject-level covariates, visit times, region-specific normative trajectories, subject random intercepts, spatial deviation maps, observation-level noise, and structured abnormality patterns in selected scenarios. Performance was summarized using estimation accuracy, deviation-map recovery, calibration, and abnormality-detection metrics.
\subsection{Data-generating mechanism}

Let $i=1,\dots,n$ index subjects, $t=1,\dots,T_i$ index visits, and $r=1,\dots,R$ index brain regions. For each subject and visit, we generate a scalar age variable $A_{it}$ and additional covariates collected in the vector $\bm{X}_{it}$. To maintain consistency with the main model, the synthetic response is generated according to
\begin{equation}
Y_{itr} = \bm{X}_{it}^{\top}\bm{\beta}_r + b_i + u_{ir} + \varepsilon_{itr},
\label{eq:sim_dgp}
\end{equation}
where:
\begin{align}
b_i &\iid \gauss(0,\sigma_b^2), \label{eq:sim_bi}\\
\bu_i = (u_{i1},\dots,u_{iR})^\top &\iid \gauss(\bm{0},\tau_u^2 \bQ(\rho)^{-1}), \label{eq:sim_ui}\\
\varepsilon_{itr} &\iid \gauss(0,\sigma^2). \label{eq:sim_eps}
\end{align}

The precision matrix $\bQ(\rho)$ is constructed from a fixed region-level adjacency matrix $\bW$ using
\begin{equation}
\bQ(\rho)=\bD-\rho \bW,
\label{eq:sim_Q}
\end{equation}
where $\bD$ is diagonal with entries $D_{rr}=\sum_{s=1}^{R}W_{rs}$. This is the standard proper conditional autoregressive construction used for structured areal spatial random effects \citep{RueHeld2005,Banerjee2014,DeOliveira2012}. The use of a proper conditional autoregressive (CAR) form is particularly appropriate here because the simulation concerns region-indexed structural measurements rather than continuous Euclidean spatial fields.

\subsubsection{Covariates and visit structure}

In the baseline simulation design, each subject contributes between two and five visits. Let $T_i$ be generated from a discrete distribution on $\{2,3,4,5\}$, or fixed at a common value in scenarios intended to isolate other features of the design. Baseline age is generated as $A_{i1} \iid \text{Uniform}(60,85),$
and follow-up visit times are generated by adding irregular increments: $A_{it} = A_{i1} + \Delta_{it}$, $t=2,\dots,T_i,$
where $\Delta_{it}$ is generated from an increasing sequence with random perturbations to mimic unequally spaced visits. A binary sex indicator $S_i \iid \text{Bernoulli}(0.5)$ is included as a time-invariant covariate. The resulting design vector takes the form $\bm{X}_{it} = (1, A_{it}, S_i)^\top$
in the baseline linear simulation. In nonlinear scenarios, the age term is replaced by a spline basis expansion.

\subsubsection{Region-specific normative coefficients}

For each region $r$, let $\bm{\beta}_r = (\beta_{0r}, \beta_{1r}, \beta_{2r})^\top.$
The intercept $\beta_{0r}$ represents the normative baseline level of region $r$, the coefficient $\beta_{1r}$ represents the age effect, and $\beta_{2r}$ represents the sex effect. To mimic heterogeneous regional trajectories, we generate region-specific coefficients as
% \begin{align*}
$\beta_{0r} \iid \gauss(\mu_0,\sigma_0^2),$ $\beta_{1r} \iid \gauss(\mu_1,\sigma_1^2)$, $\beta_{2r} \iid \gauss(\mu_2,\sigma_2^2)$,
% \end{align*}
with hyperparameters chosen so that the resulting trajectories resemble plausible variation in region-level brain structure. This random-coefficient construction is useful because it prevents the simulation from collapsing to a trivial common-trajectory setting and forces all competing models to estimate genuinely region-varying normative functions.

\subsubsection{Subject-specific deviation maps}

The vector $\bu_i$ in \eqref{eq:sim_ui} represents the individualized latent deviation map for subject $i$. This component is central to the simulation because it is precisely the object that the proposed model is intended to recover. Spatial smoothness is controlled by the pair $(\tau_u^2,\rho)$, where $\tau_u^2$ governs the marginal scale and $\rho$ governs the strength of adjacency-based dependence. When $\rho$ is close to zero, deviations are weakly structured across regions. As $\rho$ increases, neighboring regions become more strongly correlated, which reflects the empirically motivated expectation that neuroanatomical abnormalities often appear in spatially coherent regional patterns rather than as independent region-wise perturbations \citep{Huertas2017,Mansour2025}.

\subsubsection{Residual variation}

The residual term $\varepsilon_{itr}$ captures within-visit measurement noise and unexplained variability. In the main simulation design, $\sigma^2$ is chosen to yield moderate signal-to-noise ratios. Additional scenarios vary $\sigma^2$ to assess robustness under low and high residual noise.

\subsection{Abnormal subgroups and deviation signal injection}
\label{sec:abnormal_signal}
To evaluate abnormality-detection performance, we introduce a subgroup of non-reference subjects whose responses deviate systematically from the normative distribution. Let $G_i \in \{0,1\}$ indicate whether subject $i$ belongs to the abnormal subgroup. For $G_i=0$, data are generated exactly from \eqref{eq:sim_dgp}. For $G_i=1$, we modify the data-generating process by adding a region-specific abnormality signal:
\begin{equation}
Y_{itr}^{\ast} = Y_{itr} + \delta_r,
\label{eq:abnormal_signal}
\end{equation}
where $\delta_r$ is nonzero only for a prespecified subset of abnormal regions $\calA \subset \{1,\dots,R\}$. Several patterns of abnormality are considered:
\begin{enumerate}
    \item Localized abnormality: only a small contiguous cluster of regions is shifted,
    \item Diffuse abnormality: a larger distributed set of regions is shifted,
    \item Mild abnormality: $\delta_r$ is small in magnitude,
    \item Severe abnormality: $\delta_r$ is large in magnitude.
\end{enumerate}

This setup directly evaluates whether the proposed posterior deviation scores can recover abnormal patterns when spatially structured signal is present. It also permits region-level sensitivity and specificity calculations.

\subsection{Simulation scenarios}

The simulation considered six scenarios, each designed to isolate a feature commonly encountered in longitudinal neuroimaging data. In the no-spatial-dependence scenario, $\rho=0$, so the regional deviation vector was uncorrelated across regions. This setting served as a negative control and assessed whether the proposed model imposed unnecessary cost when no spatial signal was present. In the moderate-spatial-dependence scenario, $\rho$ was set to a moderate positive value, generating regional clustering in deviation maps. In the strong-spatial-dependence scenario, $\rho$ was set to a larger value, inducing highly coherent regional deviation patterns and representing settings in which spatial borrowing should be most useful.

The remaining scenarios evaluated robustness to longitudinal and mean-structure complications. In the unequal-visit scenario, the number of visits varied across subjects, mimicking unbalanced longitudinal follow-up. In the missing-follow-up scenario, some later visits were removed under a missing-at-random mechanism depending on observed age and prior regional burden. In the nonlinear-trajectory scenario, the true age effect was nonlinear. For example, data could be generated as
\[
Y_{itr}
=
\beta_{0r}
+
\beta_{1r}A_{it}
+
\beta_{2r}A_{it}^{2}
+
\beta_{3r}S_i
+
b_i
+
u_{ir}
+
\varepsilon_{itr}.
\]
In this nonlinear scenario, the coefficient vector is extended from the baseline linear form so that $\beta_{2r}$ indexes the quadratic age effect and $\beta_{3r}$ indexes the sex effect. This scenario evaluated the consequences of mean-function misspecification.

Structured abnormality patterns were introduced using the mechanism in \eqref{eq:abnormal_signal}. The same localized, diffuse, mild, and severe abnormality patterns were used across relevant scenarios, allowing region-level sensitivity, specificity, and deviation-map recovery to be evaluated against known truth.
\subsection{Competing methods}

The proposed model was compared with benchmark methods chosen to isolate the contributions of longitudinal and spatial structure. The benchmarks were not intended to exhaust all possible normative modeling approaches. Instead, they represented the main simplifications commonly used in neuroimaging applications: region-wise cross-sectional normative modeling, longitudinal normative modeling without spatial dependence, and the proposed joint longitudinal spatial formulation.

The first benchmark was an independent cross-sectional normative model. It ignored both repeated measurements and spatial dependence:
\begin{equation}
Y_{itr}
=
\bm{X}_{it}^{\top}\bm{\beta}_r
+
\varepsilon_{itr}.
\label{eq:benchmark_cs}
\end{equation}
All observations were treated as independent. This benchmark reflects the basic structure of region-wise normative modeling, in which separate predictive models are fit to individual brain regions or locations and subject-specific deviations are computed relative to the fitted normative distribution. This strategy has been central to neuroimaging normative modeling since its early use for individual-level inference in heterogeneous clinical cohorts \citep{Marquand2016,Marquand2019}. More recent scalable Bayesian and distributional extensions, including hierarchical Bayesian regression and warped Bayesian linear regression, retain the same basic objective of estimating individualized deviations from a reference distribution while improving multi-site calibration, flexibility, and scalability \citep{Kia2020,Fraza2021,Rutherford2022}.

The second benchmark was a longitudinal non-spatial normative model. It accounted for repeated measurements through a subject-specific random intercept but ignored spatial dependence:
\begin{equation}
Y_{itr}
=
\bm{X}_{it}^{\top}\bm{\beta}_r
+
b_i
+
\varepsilon_{itr},
\label{eq:benchmark_long}
\end{equation}
where $b_i\sim\mathcal{N}(0,\sigma_b^2)$. This model corresponds to the standard mixed-effects representation for longitudinal data \citep{LairdWare1982}. It is included because recent neuroimaging work has emphasized that normative modeling must move beyond single-occasion cross-sectional comparisons when the scientific question concerns within-person change, disease progression, or developmental trajectory \citep{Gaiser2024,Buckova2025,Verdi2024}. However, these longitudinal normative modeling strategies generally focus on temporal change or scanner harmonization and do not directly model a subject-specific spatial deviation process over brain regions. This benchmark therefore isolates the gain attributable to longitudinal structure alone.

The third method was the proposed Bayesian longitudinal spatial normative model:
\begin{equation}
Y_{itr}
=
\bm{X}_{it}^{\top}\bm{\beta}_r
+
b_i
+
u_{ir}
+
\varepsilon_{itr},
\label{eq:benchmark_full}
\end{equation}
with $\bu_i\sim\gauss(\bm{0},\tau_u^2\bQ(\rho)^{-1})$. The spatial component is modeled as a Gaussian Markov random field, drawing on the conditional autoregressive and GMRF literature for spatially indexed data \citep{Besag1974,RueHeld2005,Banerjee2014}. This formulation is motivated by the fact that structural brain deviations are often spatially coherent rather than independent across regions. Recent spatial normative modeling work has made a similar point at the level of high-resolution brain structure by showing that spatial representations can improve the characterization of individual deviation maps \citep{Mansour2025}. The distinction in the present work is that spatial dependence is embedded directly within a longitudinal Bayesian hierarchical model, so the subject-specific deviation map is estimated jointly with repeated-measure dependence and normative covariate effects.

\subsection{Performance measures and Monte Carlo implementation}

The simulation used the same prior family as the estimation model. For regression coefficients, we assigned weakly informative Gaussian priors $\bm{\beta}_r\sim\gauss(\bm{0},\sigma_\beta^2\bI_p)$. For the scale parameters $\sigma$, $\sigma_b$, and $\tau_u$, we used half-Cauchy priors centered at zero with fixed scale hyperparameters. The choice of half-Cauchy priors follows the hierarchical modeling literature, where they have been recommended as weakly informative defaults for variance components and as practical alternatives to inverse-Gamma priors in multilevel models \citep{Gelman2006,PolsonScott2012}. For the spatial dependence parameter $\rho$, we used a prior supported on the admissible interval implied by $\bQ(\rho)$; a transformed Beta prior was used when a bounded support was required.

For each Monte Carlo replicate and each fitted method, performance was summarized using estimation accuracy, recovery of subject-specific deviation maps, calibration, abnormal-region detection, and predictive fit. Let $\mu_{itr}^{\text{true}}$ denote the true normative mean used in data generation, and let $\widehat{\mu}_{itr}$ denote the fitted posterior predictive mean. Normative mean accuracy was evaluated using
\[
\text{Bias}
=
\frac{1}{N_{\text{obs}}}
\sum_{i,t,r}
\left(
\widehat{\mu}_{itr}
-
\mu_{itr}^{\text{true}}
\right),
\qquad
\text{MSE}
=
\frac{1}{N_{\text{obs}}}
\sum_{i,t,r}
\left(
\widehat{\mu}_{itr}
-
\mu_{itr}^{\text{true}}
\right)^2.
\]
Recovery of the subject-specific spatial deviation map was evaluated using
\[
\text{Map-MSE}
=
\frac{1}{nR}
\sum_{i=1}^{n}
\left\|
\widehat{\bu}_i
-
\bu_i^{\text{true}}
\right\|_2^2,
\]
where $\bu_i^{\text{true}}$ denotes the true latent deviation map and $\widehat{\bu}_i$ denotes its posterior mean estimate. This metric directly evaluates the primary object of methodological interest.

Calibration was assessed in two ways. First, for scalar targets, we calculated empirical coverage of nominal $95\%$ posterior credible intervals for the normative mean and, where appropriate, for components of the subject-specific deviation map. Second, for observations generated from the reference subgroup, we evaluated whether the standardized deviation scores were centered near zero with variance close to one. Specifically, we summarized
\[
\frac{1}{N_{\text{ref}}}
\sum_{(i,t,r)\in\calN}
Z_{itr},
\qquad
\frac{1}{N_{\text{ref}}-1}
\sum_{(i,t,r)\in\calN}
\left(
Z_{itr}-\bar{Z}
\right)^2.
\]
We also recorded the empirical proportion of $|Z_{itr}|>1.96$ in the reference subgroup as a practical tail-calibration check. Proper calibration is essential in normative modeling because deviation scores are interpreted probabilistically rather than only descriptively \citep{Marquand2016,Rutherford2022}.

For abnormal subjects, region-level detection was assessed by thresholding posterior deviation scores and comparing the flagged set with the true abnormal set $\calA$. We reported sensitivity, specificity, positive predictive value, and the area under the receiver operating characteristic curve when continuous deviation scores were used. Predictive fit was also summarized using mean absolute error and mean standardized log-loss when useful, following the broader normative-modeling literature in which both central tendency and distributional fit are relevant \citep{Mansour2025}.

For each scenario, $M$ independent datasets were generated, with $M$ chosen large enough to stabilize Monte Carlo summaries. Within each replicate, all competing methods were fit to the same generated dataset. Posterior inference was carried out using the same computational backend across Bayesian models, with common convergence diagnostics including trace inspection, $\widehat{R}$, and effective sample size. Summary results were reported as Monte Carlo means together with Monte Carlo standard errors across replicates.

To avoid optimistic evaluation, abnormality-detection metrics were computed relative to the known generating truth rather than to the fitted model itself. Likewise, when the simulation used a distinct reference subgroup, fitting of the normative model was restricted to that subgroup before deviation scores were evaluated on held-out or abnormal subjects.

\subsection{Simulation results}

Each Monte Carlo dataset contained 120 subjects measured over 20 brain regions, with scenario-specific variation in spatial dependence, visit structure, missing follow-up, and mean-function form. The primary inferential target was recovery of the subject-specific deviation map. Results are reported in Tables~\ref{tab:accuracy}--\ref{tab:calibration} and Figures~\ref{fig:accuracy}--\ref{fig:calibration}.

\subsubsection{Estimation accuracy and deviation-map recovery}

Table~\ref{tab:accuracy} and Figure~\ref{fig:accuracy} show that the Bayesian longitudinal spatial model achieved the lowest deviation-map mean squared error in every scenario. In the no-spatial-dependence setting, the proposed model produced a Map MSE of 0.352, compared with 0.690 for the longitudinal non-spatial model and 0.847 for the independent cross-sectional model. Thus, introducing spatial structure did not reduce performance when the generating mechanism contained no regional dependence.

The gains became larger when regional dependence was present. Under moderate spatial dependence, Map MSE decreased from 0.928 for the independent model and 0.718 for the longitudinal non-spatial model to 0.385 for the Bayesian longitudinal spatial model. Under strong spatial dependence, the corresponding values were 1.337, 0.930, and 0.604. Relative to the independent benchmark, the proposed model reduced map reconstruction error by approximately 55\% under both moderate and strong spatial dependence.
The same pattern appeared under unequal visit schedules, missing follow-up, and nonlinear mean trajectories. In the variable-visit scenario, the longitudinal non-spatial model had slightly smaller global MSE than the Bayesian longitudinal spatial model, 1.083 versus 1.093, but the proposed model had substantially lower Map MSE, 0.411 compared with 0.736. Under missing follow-up and nonlinear mean structure, the proposed model again produced the smallest Map MSE while maintaining stable global estimation accuracy.
Across all scenarios, mean bias for the Bayesian longitudinal spatial model remained close to zero, ranging from 0.002 to 0.005. The independent model showed larger positive bias, ranging from 0.029 to 0.058, while the longitudinal non-spatial model ranged from 0.018 to 0.043.
\begingroup
\small
\setlength{\tabcolsep}{3pt}
\begin{longtable}[]{@{}llccc@{}}
\caption{Simulation performance across the six data-generating scenarios. Lower values of mean squared error (MSE) and deviation-map reconstruction error (Map MSE) indicate improved recovery of the underlying latent deviation structure.}
\label{tab:accuracy}\tabularnewline
\toprule\noalign{}
\textbf{Simulation scenario} & \textbf{Model} & \textbf{Mean bias} & \textbf{MSE} & \textbf{Map MSE} \\
\midrule\noalign{}
\endfirsthead

\toprule\noalign{}
\textbf{Simulation scenario} & \textbf{Model} & \textbf{Mean bias} & \textbf{MSE} & \textbf{Map MSE} \\
\midrule\noalign{}
\endhead

\bottomrule\noalign{}
\endlastfoot

No spatial dependence
& Independent cross-sectional  & 0.031 & 1.145 & 0.847 \\
& Longitudinal non-spatial  & 0.023 & 1.036 & 0.690 \\
& Bayesian longitudinal spatial  & 0.002 & 0.887 & 0.352 \\

Moderate spatial dependence
& Independent cross-sectional  & 0.031 & 1.177 & 0.928 \\
& Longitudinal non-spatial  & 0.022 & 1.027 & 0.718 \\
& Bayesian longitudinal spatial  & 0.002 & 0.908 & 0.385 \\

Strong spatial dependence
& Independent cross-sectional  & 0.030 & 1.391 & 1.337 \\
& Longitudinal non-spatial & 0.018 & 1.074 & 0.930 \\
& Bayesian longitudinal spatial  & 0.002 & 1.063 & 0.604 \\

Variable visit schedule
& Independent cross-sectional & 0.030 & 1.392 & 1.136 \\
& Longitudinal non-spatial  & 0.019 & 1.083 & 0.736 \\
& Bayesian longitudinal spatial  & 0.002 & 1.093 & 0.411 \\

Missing follow-up
& Independent cross-sectional  & 0.029 & 1.385 & 1.119 \\
& Longitudinal non-spatial  & 0.019 & 1.091 & 0.743 \\
& Bayesian longitudinal spatial  & 0.002 & 1.086 & 0.410 \\

Nonlinear mean structure
& Independent cross-sectional & 0.058 & 1.242 & 0.973 \\
& Longitudinal non-spatial  & 0.043 & 1.088 & 0.758 \\
& Bayesian longitudinal spatial  & 0.005 & 0.971 & 0.409 \\

\end{longtable}
\endgroup

\begin{figure}[ht]
\centering
\includegraphics[width=5in,height=\textheight]{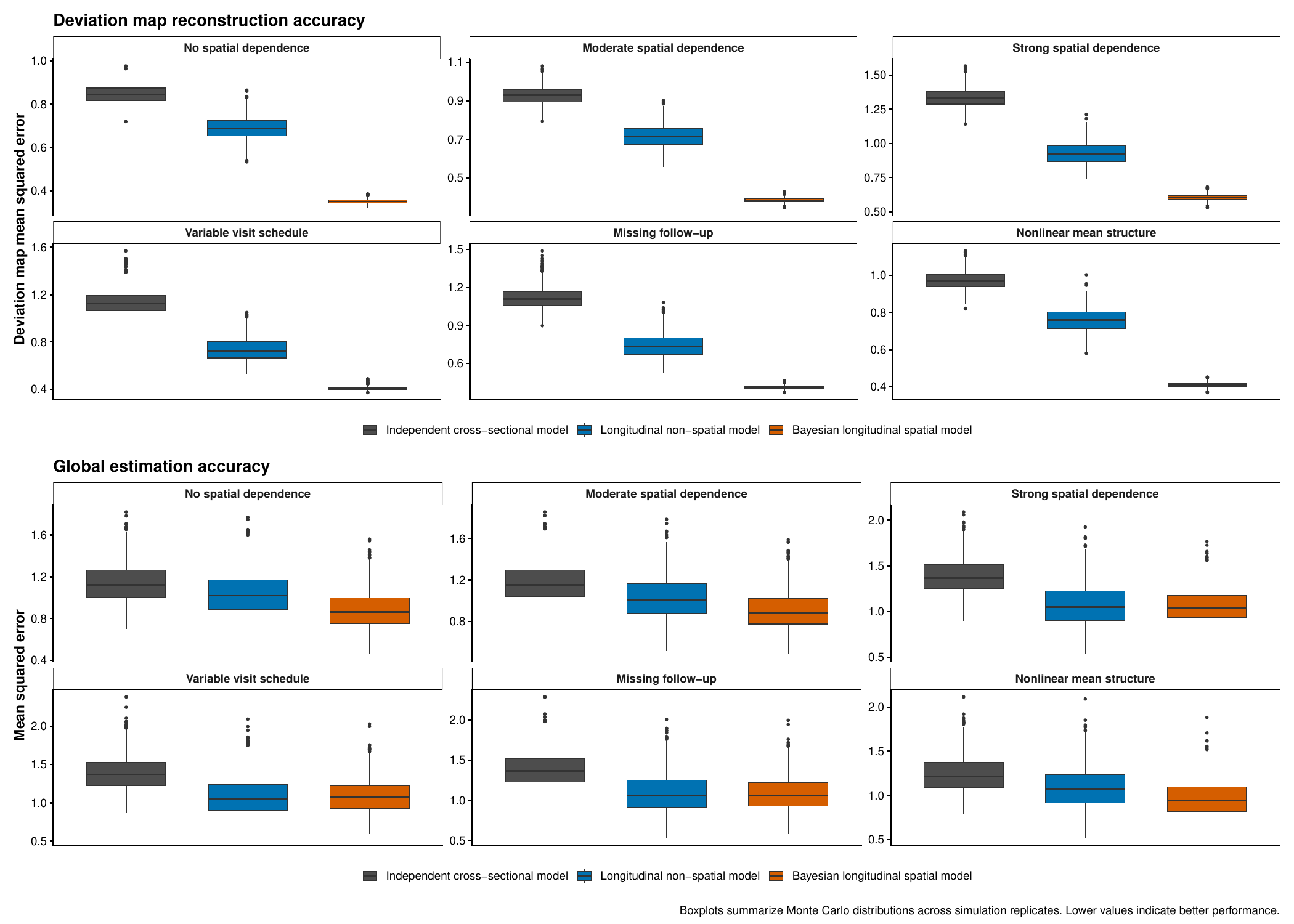}
\caption{Estimation accuracy and deviation-map reconstruction across simulation scenarios. Boxplots summarize Monte Carlo distributions across simulation replicates. Lower values indicate better performance.}
\label{fig:accuracy}
\end{figure}

 These results show that joint modeling of repeated measurements and regional dependence improves recovery of individualized deviation structure without sacrificing normative mean estimation accuracy. The pattern is consistent with the motivating premise that spatial structure is most useful when subject-specific deviations are regionally coherent \citep{Marquand2016,Rutherford2022,Mansour2025}.

\subsubsection{Calibration of standardized deviation scores}

Table~\ref{tab:calibration} and Figure~\ref{fig:calibration} summarize calibration of standardized deviation scores. The Bayesian longitudinal spatial model produced z-score means closest to zero across all scenarios, ranging from $-0.004$ to $-0.001$. The independent model ranged from $-0.046$ to $-0.022$, while the longitudinal non-spatial model ranged from $-0.035$ to $-0.015$.
\begingroup
\small
\setlength{\tabcolsep}{3pt}
\renewcommand{\arraystretch}{1.08}

\begin{longtable}[]{@{}p{5.0cm}p{5.5cm}ccc@{}}
\caption{Calibration performance across simulation scenarios. Target values are 0 for z-score mean, 1 for z-score variance, and 0.05 for tail probability.}
\label{tab:calibration}\tabularnewline

\toprule\noalign{}
\textbf{Scenario} &
\textbf{Model} &
\textbf{$Z$ mean} &
\textbf{$Z$ var.} &
\textbf{Tail prob.} \\
\midrule\noalign{}
\endfirsthead

\toprule\noalign{}
\textbf{Scenario} &
\textbf{Model} &
\textbf{$Z$ mean} &
\textbf{$Z$ var.} &
\textbf{Tail prob.} \\
\midrule\noalign{}
\endhead

\bottomrule\noalign{}
\endlastfoot

No spatial dependence
& Independent cross-sectional & -0.025 & 0.981 & 0.044 \\
& Longitudinal non-spatial  & -0.020 & 0.920 & 0.037 \\
& Bayesian longitudinal spatial  & -0.002 & 0.966 & 0.046 \\

Moderate spatial dependence
& Independent cross-sectional  & -0.025 & 0.981 & 0.044 \\
& Longitudinal non-spatial  & -0.018 & 0.904 & 0.036 \\
& Bayesian longitudinal spatial  & -0.002 & 0.966 & 0.046 \\

Strong spatial dependence
& Independent cross-sectional & -0.023 & 0.982 & 0.045 \\
& Longitudinal non-spatial  & -0.015 & 0.864 & 0.031 \\
& Bayesian longitudinal spatial  & -0.002 & 0.965 & 0.046 \\

Variable visit schedule
& Independent cross-sectional  & -0.023 & 0.980 & 0.047 \\
& Longitudinal non-spatial  & -0.015 & 0.886 & 0.036 \\
& Bayesian longitudinal spatial  & -0.002 & 0.960 & 0.045 \\

Missing follow-up
& Independent cross-sectional  & -0.022 & 0.981 & 0.047 \\
& Longitudinal non-spatial  & -0.015 & 0.889 & 0.036 \\
& Bayesian longitudinal spatial  & -0.001 & 0.964 & 0.046 \\

Nonlinear mean structure
& Independent cross-sectional & -0.046 & 0.965 & 0.042 \\
& Longitudinal non-spatial  & -0.035 & 0.890 & 0.034 \\
& Bayesian longitudinal spatial  & -0.004 & 0.955 & 0.044 \\

\end{longtable}
\endgroup
\begin{figure}[ht]
\centering
\includegraphics[width=4in,height=\textheight]{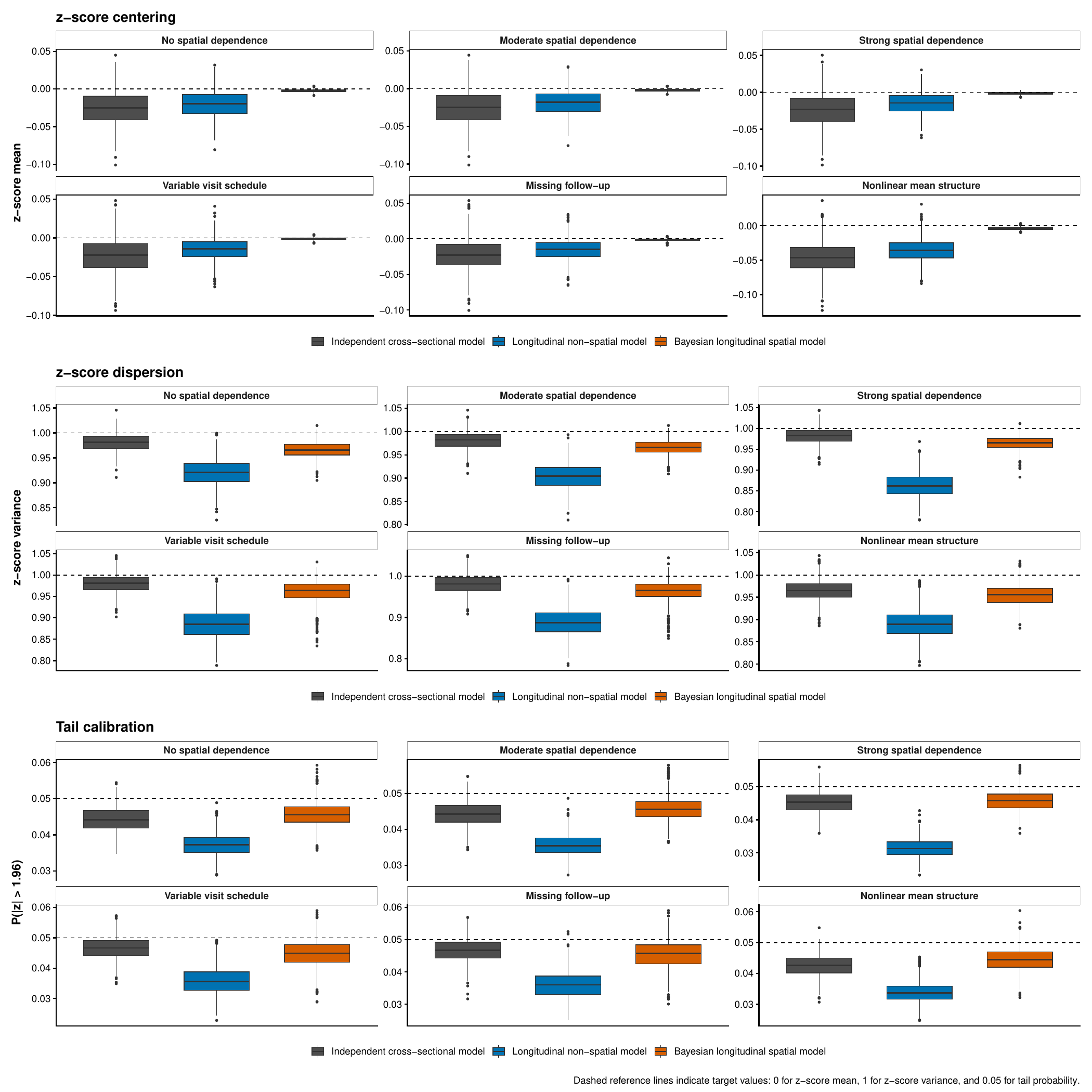}
\caption{Calibration under varying spatial and longitudinal data-generating mechanisms. Dashed reference lines indicate target values: 0 for z-score mean, 1 for z-score variance, and 0.05 for tail probability.}
\label{fig:calibration}
\end{figure}
The proposed model showed mild variance shrinkage, with z-score variance between 0.955 and 0.966. The independent model remained closer to the nominal target of 1, whereas the longitudinal non-spatial model was more strongly under-dispersed, with variance ranging from 0.864 to 0.920. Tail calibration remained stable for the Bayesian longitudinal spatial model, with empirical probabilities of $|Z|>1.96$ between 0.044 and 0.046, close to the nominal target of 0.05. The longitudinal non-spatial model produced smaller tail probabilities, especially under stronger spatial dependence.

The calibration results support the reconstruction findings. The proposed model produced substantially lower deviation-map error while preserving centered standardized scores and stable tail behavior. This operating behavior is important for individualized structural brain mapping, where the target is not only prediction but recovery of a spatial pattern of subject-level anatomical departure.
\addtolength{\textheight}{-.2in}%
\section{Application to OASIS-3 structural MRI}\label{sec-oasis}

We applied the proposed model to longitudinal structural MRI data obtained from OASIS-3 \citep{LaMontagne2019}, a multimodal neuroimaging, clinical, cognitive, and biomarker resource for normal aging and Alzheimer disease. The OASIS-3 project includes repeated MRI sessions, FreeSurfer \citep{fischl2004automatically} derived structural summaries, cognitive assessments, and clinical information collected through the Washington University Knight Alzheimer Disease Research Center. The application was designed to evaluate whether the proposed longitudinal spatial model improves subject-level normative deviation estimation in a real structural neuroimaging setting.

The analysis used FreeSurfer-derived regional measures from subjects with at least two MRI visits. We focused on a set of anatomically relevant cortical thickness and subcortical volume measures, including bilateral hippocampus, amygdala, entorhinal cortex, parahippocampal cortex, temporal pole, inferior temporal cortex, middle temporal cortex, posterior cingulate, and precuneus. Regional measures were standardized within region prior to modeling so that cortical thickness and volume features could be analyzed on a common scale. Age at MRI visit and FreeSurfer processing version were included as covariates. The model comparison used the same three-model structure as the simulation study: an independent cross-sectional model, a longitudinal non-spatial model with a subject random intercept, and the proposed Bayesian subject-specific spatial model.

The real-data analysis included 1,902 region-level observations for each selected region. The proposed Bayesian model was fit using four Markov chains with 1,000 warmup iterations and 1,000 post-warmup iterations per chain. The key posterior scale parameters were well identified. The residual standard deviation was estimated as 0.491, the subject-level random-intercept standard deviation as 0.517, the spatial deviation scale as 0.928, and the FreeSurfer-version scale as 0.305. The relatively large posterior mean of $\tau_u$ indicates substantial subject-specific spatial variation across brain regions after accounting for age, repeated observations, and processing version. This is consistent with the motivating premise of the proposed model: individual structural brain deviations are not purely independent regional residuals, but have a structured spatial component.

\subsection{Model comparison}
The proposed Bayesian subject-specific spatial model provided the strongest fit among the three competing approaches (Table~\ref{tab:oasis_model_comparison}; Figure~\ref{fig:oasis_model_comparison}). The independent cross-sectional model had an RMSE of 0.926 and mean absolute deviation of 0.716. Adding a subject random intercept improved fit, reducing RMSE to 0.774 and mean absolute deviation to 0.597. The Bayesian subject-specific spatial model produced the largest improvement, reducing RMSE to 0.423 and mean absolute deviation to 0.306.
These gains were not small changes in numerical fit. Relative to the independent cross-sectional model, the proposed model reduced RMSE by 54.3\% and mean absolute deviation by 57.3\%. Relative to the longitudinal non-spatial model, the proposed model reduced RMSE by 45.3\% and mean absolute deviation by 48.7\%.
\begin{longtable}[]{@{}p{5.9cm}ccccc@{}}
\caption{Model comparison in the OASIS-3 structural MRI application. Lower residual standard deviation, mean absolute deviation, and RMSE indicate improved model fit.}
\label{tab:oasis_model_comparison}\tabularnewline

\toprule\noalign{}
\textbf{Model} &
\textbf{Res.\ SD} &
\textbf{MAD} &
\textbf{RMSE} &
\textbf{z-var} &
\textbf{Tail prob.} \\
\midrule\noalign{}
\endfirsthead

\toprule\noalign{}
\textbf{Model} &
\textbf{Res.\ SD} &
\textbf{MAD} &
\textbf{RMSE} &
\textbf{z-var} &
\textbf{Tail prob.} \\
\midrule\noalign{}
\endhead

\bottomrule\noalign{}
\endlastfoot

Independent cross-sectional
& 0.926 & 0.716 & 0.926 & 1.000 & 0.0503 \\

Longitudinal non-spatial
& 0.774 & 0.597 & 0.774 & 1.000 & 0.0482 \\

Bayesian subject-specific spatial
& 0.423 & 0.306 & 0.423 & 0.743 & 0.0314 \\

\end{longtable}
\begin{figure}[H]
\centering
\includegraphics[width=4in,height=\textheight]{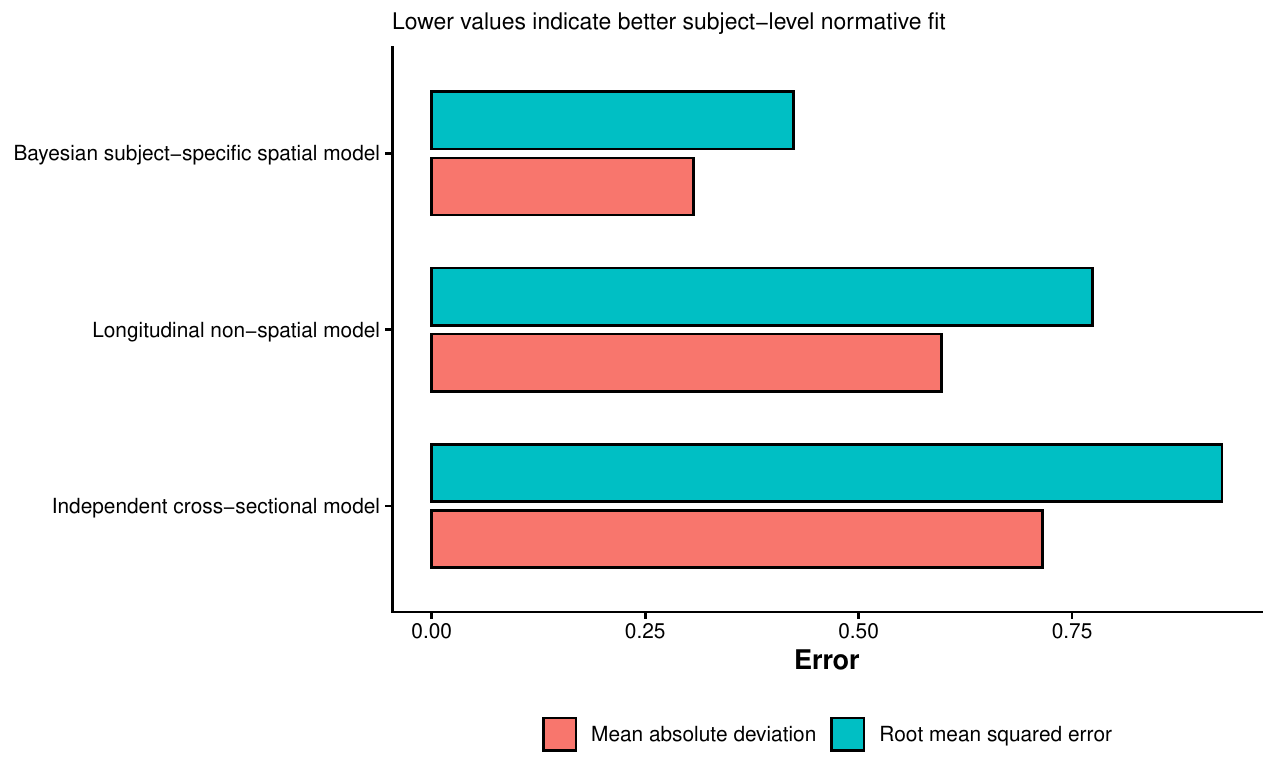}
\caption{Model comparison in the OASIS-3 real-data application. The Bayesian subject-specific spatial model produced the lowest RMSE and mean absolute deviation, indicating improved subject-level normative fit relative to both benchmark models.}
\label{fig:oasis_model_comparison}
\end{figure}
 Thus, repeated-measure modeling improved the normative fit, but the largest gain came from modeling subject-specific spatial deviation across regions. This pattern agrees with the simulation results and supports the central methodological claim of the paper.

The standardized deviation summaries also showed an important difference between fit and calibration. The independent and longitudinal non-spatial models had z-score variances close to 1 by construction in the residual scale used for comparison, with tail probabilities of 0.0503 and 0.0482, respectively. The Bayesian spatial model had a smaller z-score variance of 0.743 and a tail probability of 0.0314. This under-dispersion reflects posterior shrinkage from the spatial hierarchical structure. In this application, the shrinkage is expected because the model borrows information across anatomically related regions and repeated visits. The result should therefore be interpreted as more conservative extreme-deviation calling, not as failure of the model. Supplementary Figures~\ref{fig:s_histograms}-\ref{fig:s_traceplots} provide additional calibration, QQ, and observed-versus-fitted diagnostics.

\subsection{Posterior predictive assessment}
Posterior predictive diagnostics were used to evaluate adequacy of the Bayesian longitudinal spatial normative model. Replicated datasets generated from the posterior predictive distribution showed close agreement with the observed distribution of standardized deviations, regional variability, and subject-level abnormality burden. Calibration histograms and QQ plots indicated that the posterior predictive standardized deviations remained approximately centered around the normative reference distribution while preserving anatomically structured variability across regions. Observed-versus-fitted comparisons further demonstrated that the proposed framework adequately captured both central tendency and subject-specific deviation structure across repeated visits.
Additional posterior predictive diagnostics, including calibration histograms, QQ plots, and observed-versus-fitted comparisons, are provided in Supplementary  Figures~\ref{fig:s_histograms}-\ref{fig:s_traceplots}.
\subsection{Regional and subject-level deviation patterns}
The cortical atlas analysis showed that extreme standardized deviations were concentrated in temporolimbic and posterior association regions (Figure~\ref{fig:oasis_atlas}). A complementary map of posterior mean cortical standardized deviations is provided in Supplementary Figure~\ref{fig:s_mean_cortical}. The highest cortical tail probabilities occurred in the right temporal pole, left temporal pole, right entorhinal cortex, and left entorhinal cortex. These regions had tail probabilities of 0.0778, 0.0705, 0.0547, and 0.0515, respectively. Inferior temporal cortex, posterior cingulate, parahippocampal cortex, middle temporal cortex, and precuneus also appeared among the higher-burden regions.
\begin{figure}[ht]
\centering
\includegraphics[width=4in,height=\textheight]{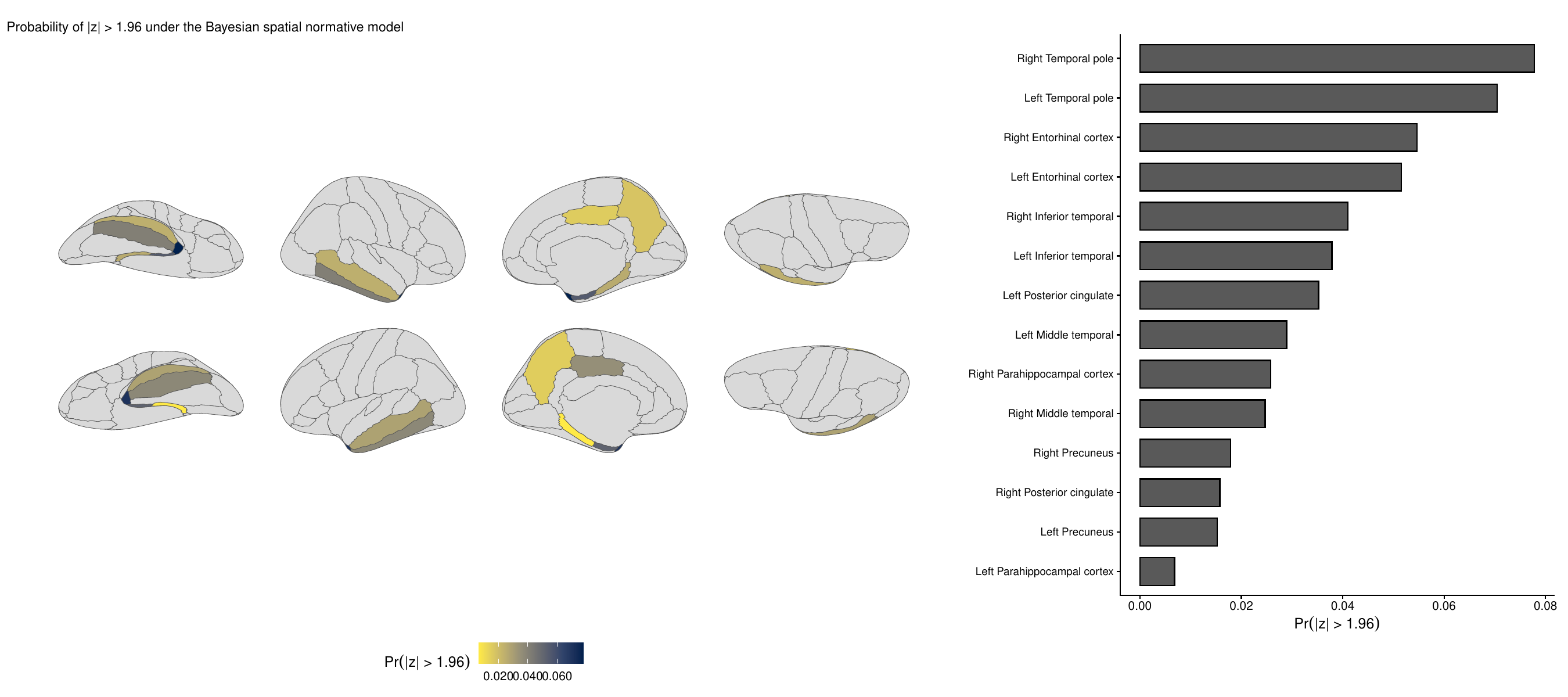}
\caption{Cortical burden of extreme standardized deviations in the OASIS-3 application. The map displays $\Pr(|Z|>1.96)$ across mapped cortical regions, with a ranked panel showing the corresponding region-level values. Highest burdens were observed in temporal pole and entorhinal regions, followed by inferior temporal, posterior cingulate, parahippocampal, middle temporal, and precuneus regions.}
\label{fig:oasis_atlas}
\end{figure}

\begin{figure}[!ht]
\centering
\includegraphics[width=3in,height=\textheight]{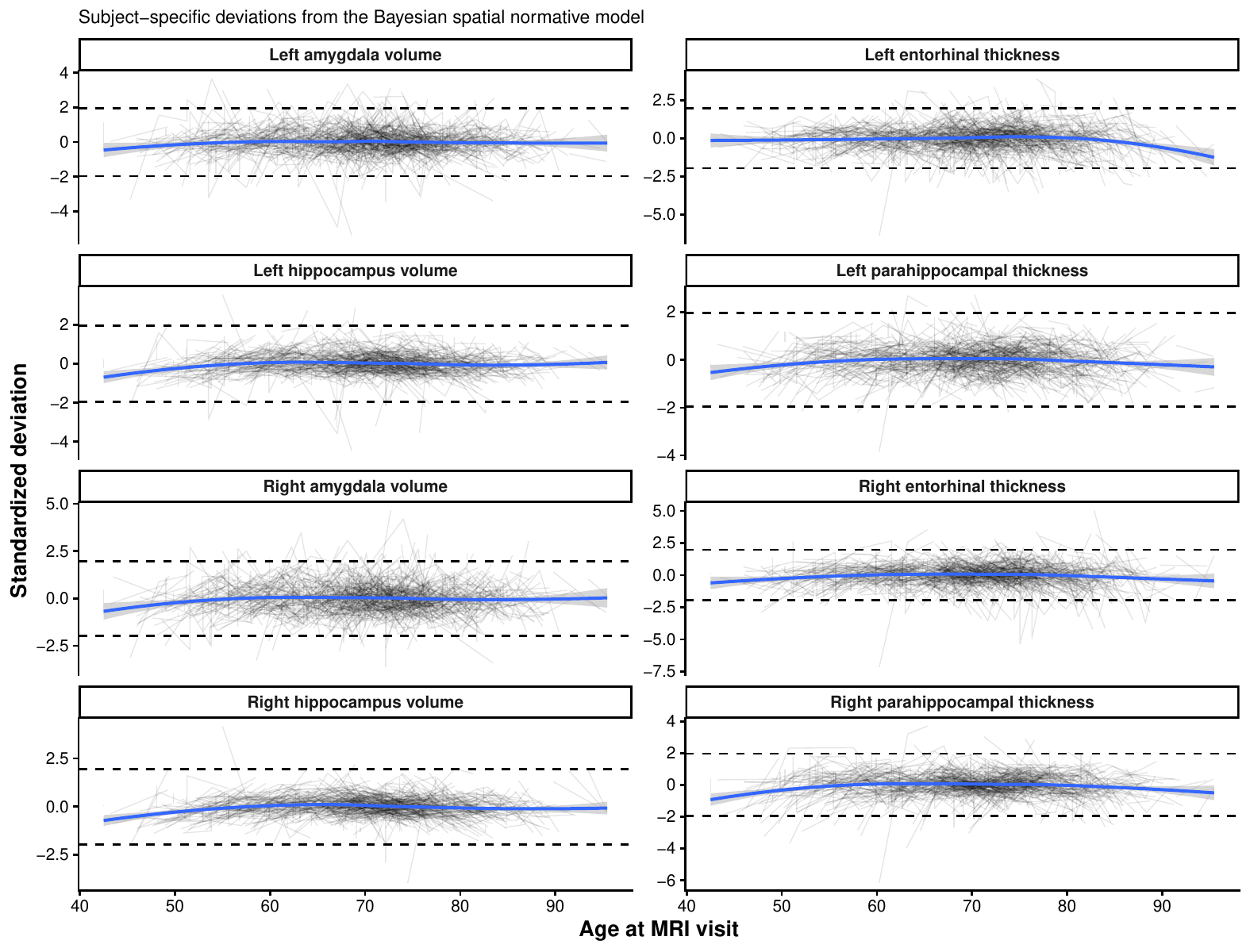}
\caption{Longitudinal standardized deviation trajectories for selected structural regions. Gray lines represent subject-specific trajectories, blue curves summarize the average smoothed trajectory, and dashed horizontal lines denote $\pm 1.96$. Most subjects remained near the normative center, while a subset showed region-specific extreme deviations.}
\label{fig:oasis_trajectories}
\end{figure}

This spatial pattern is biologically plausible. The entorhinal cortex, parahippocampal region, hippocampus, temporal pole, and inferior temporal cortex are closely related to memory and limbic-temporal networks affected early in Alzheimer-type neurodegeneration. Posterior cingulate and precuneus are also central regions in aging and Alzheimer disease research because of their involvement in default-mode and memory-related networks. The model was not given diagnostic labels as the primary driver of these maps. The emergence of these regions from subject-level deviation estimates suggests that the proposed model is capturing meaningful anatomical variation rather than only improving numerical fit.

Subcortical structures were not displayed in the Desikan-Killiany cortical atlas map, but the region-level table showed that left amygdala volume and right amygdala volume also had elevated abnormality burden, with tail probabilities of 0.0342 and 0.0331. The full ranked region summaries are provided in Supplementary Table~\ref{tab:s_top_regions}.
\begin{figure}[!ht]
\centering
\includegraphics[width=3in,height=\textheight]{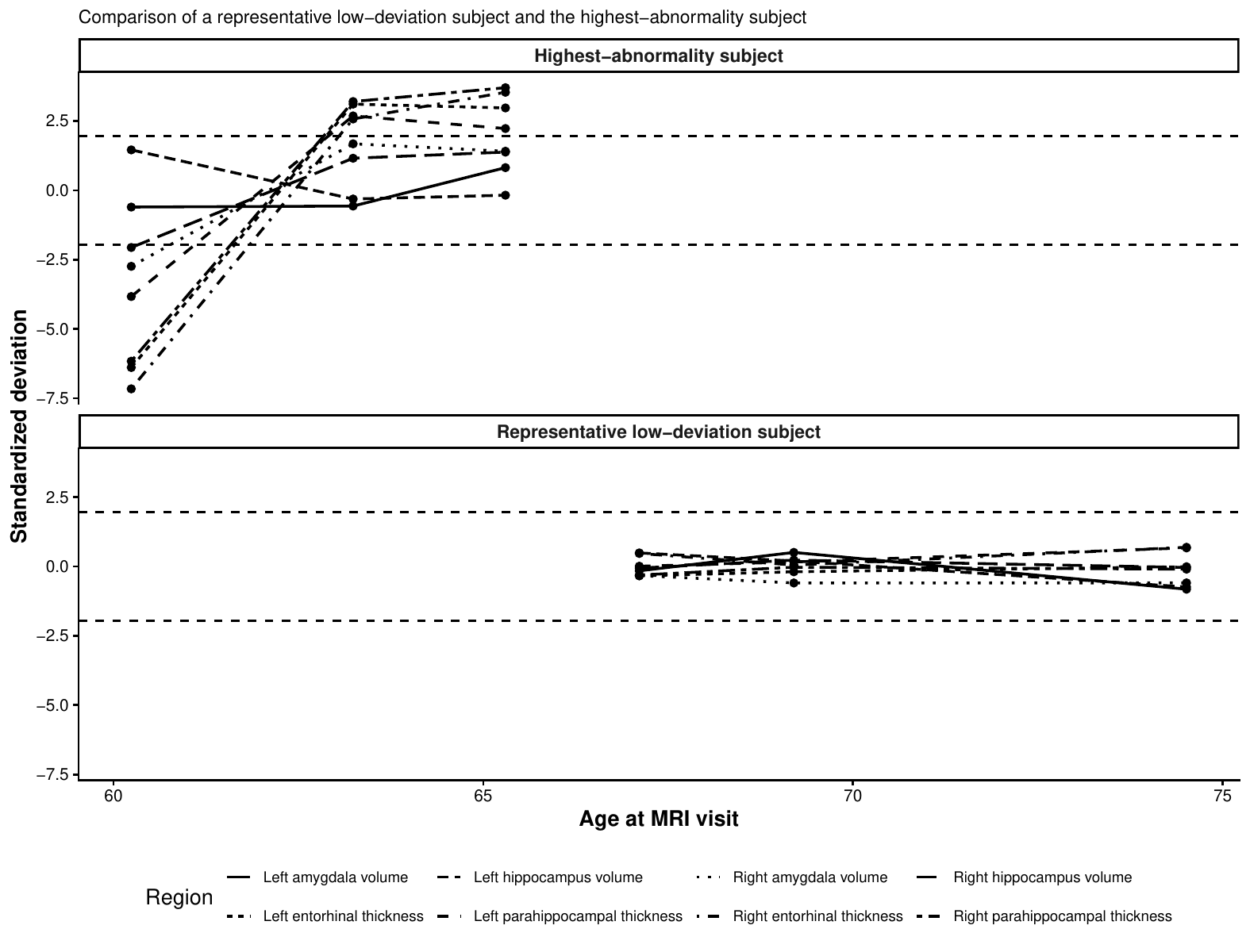}
\caption{Illustrative subject-level deviation profiles. The upper panel shows the subject with the largest mean absolute standardized deviation, while the lower panel shows a representative low-deviation subject. Dashed horizontal lines denote $\pm 1.96$. The high-abnormality profile shows large and persistent deviations across several temporal and limbic regions, whereas the low-deviation profile remains close to the normative center.}
\label{fig:oasis_case}
\end{figure} 
The longitudinal deviation trajectories provide a subject-level view of how structural deviations vary across age and region (Figure~\ref{fig:oasis_trajectories}). Across the selected medial temporal and limbic regions, most trajectories remained close to zero, indicating that the fitted model preserved a stable normative center for the majority of observations. At the same time, a subset of subjects showed deviations outside the $\pm 1.96$ reference bands, especially in hippocampal, amygdala, entorhinal, and parahippocampal measures. This is the desired behavior for a normative model: the population center remains stable while individual subjects can depart strongly from the expected structural profile.

The smoothed trajectories were generally near zero across age, with modest curvature in several regions. This suggests that the model adjusted for age-related mean trends while still allowing subject-specific departures from the fitted normative surface. The full calibration and observed-versus-fitted diagnostics are reported in Supplementary Figures~\ref{fig:s_histograms}-\ref{fig:s_qq}.

The subject-level case study illustrates how the proposed model can separate a stable low-deviation profile from a highly abnormal regional deviation profile (Figure~\ref{fig:oasis_case}). The highest-abnormality subject, OAS30303, had 57 observations across three MRI visits, mean absolute standardized deviation 3.97, maximum absolute deviation 10.5, and 77.2\% of observations exceeding $|Z|>1.96$. The subject had mean MMSE 29.7 and mean CDR total score 0, indicating that the structural deviation profile was not simply a reflection of severe measured cognitive impairment at the available visits.

This example is important because it shows the scientific value of individualized normative modeling. A subject may have strong multiregional structural deviation even when global cognitive summaries remain within a clinically normal or minimally impaired range. The proposed model therefore provides a way to identify spatially coherent anatomical deviation patterns that may be hidden by group-level summaries or by cognitive status alone. The representative low-deviation subject remained close to the normative center across selected regions and visits, demonstrating the contrast between stable normative aging and marked subject-level anatomical departure. The full distribution of subject-level abnormality burden is shown in Supplementary Figure~\ref{fig:s_subject_burden}, and full subject summaries are provided in Supplementary Table~\ref{tab:s_case_subjects}.

\section{Discussion}\label{sec-discu}

In this paper, we developed a Bayesian longitudinal spatial normative 
model for individualized brain deviation mapping from repeated 
structural neuroimaging measurements. The framework was motivated by a limitation of many existing normative approaches: 
individualized deviation is often treated implicitly as a residual 
after estimation of a normative mean model. In longitudinal 
neuroimaging settings, deviation itself may exhibit persistent spatial 
organization across anatomically related regions and repeated visits. 
The proposed formulation addresses this directly by modeling 
individualized deviation as a subject-specific latent spatial process 
within a longitudinal Bayesian framework.

The theoretical results provide the individualized deviation map with 
a formal Bayesian interpretation. Conditional on model parameters, the 
posterior distribution of the latent spatial deviation vector combines 
information across repeated visits and neighboring brain regions. Under 
squared error loss, the posterior mean is the Bayes estimator of the 
individualized deviation map. The resulting map is therefore not an 
ad hoc smoothed residual image but a posterior estimand with an 
explicit covariance structure and principled uncertainty 
quantification. This distinction is central to the methodological 
contribution of the paper.

Simulation results confirmed that jointly modeling longitudinal and 
spatial structure consistently improved deviation-map reconstruction 
relative to both benchmarks, with the strongest gains under moderate 
and strong spatial dependence. Critically, the model remained stable 
when spatial dependence was weak, and maintained well-calibrated 
deviation scores across all scenarios. In normative modeling, deviation scores are interpreted relative to a covariate-adjusted reference distribution rather than as ordinary prediction residuals \citep{Marquand2016,Rutherford2022,
RutherfordEvidence2022}. The benchmark comparisons further showed 
that subject random intercepts alone cannot represent how deviations 
are distributed across brain regions: longitudinal dependence is 
necessary but insufficient when the inferential target is the regional 
organization of individualized anatomical departure.

In the OASIS-3 application, the largest improvement occurred after 
introducing the subject-specific spatial deviation process. The 
posterior mean of the spatial deviation scale exceeded the residual 
scale, indicating that structured regional heterogeneity is an 
important component of subject-level variation. The Bayesian spatial 
model produced lower z-score variance and fewer extreme deviations 
than the benchmarks, reflecting posterior shrinkage from borrowing 
information across anatomically related regions and repeated visits. 
Such shrinkage may be desirable in individualized neuroimaging analysis 
because isolated extreme residuals may reflect measurement instability 
rather than coherent anatomical abnormality. Under the proposed 
framework, observations are flagged as extreme only after accounting 
for age, repeated measurements, processing version, and regional 
dependence.

The proposed framework is related to several existing directions in 
normative modeling and Bayesian neuroimaging but addresses a different 
inferential objective. Normative modeling established covariate-adjusted 
reference distributions for individualized inference 
\citep{Marquand2016,Rutherford2022}. Hierarchical Bayesian extensions 
improved multi-site calibration \citep{Kia2020,bayer2022accommodating}. 
Distributional and warped approaches extended the framework beyond 
Gaussian assumptions \citep{Fraza2021}. Bayesian spatial neuroimaging 
models demonstrated structured borrowing across brain regions 
\citep{Huertas2017,Mejia2022}, and recent spectral normative models 
showed that spatial representations support high-resolution deviation 
mapping \citep{Mansour2025}. The present work is complementary to 
these developments. Its emphasis is joint Bayesian estimation of longitudinal dependence and subject-specific spatial deviation, with individualized deviation treated as the primary latent inferential target rather than spectral reconstruction or atlas-free spatial querying.

The regional findings in OASIS-3 were anatomically plausible. The 
largest extreme-deviation burdens were observed in the temporal pole, 
entorhinal cortex, inferior temporal cortex, posterior cingulate, 
parahippocampal cortex, middle temporal cortex, precuneus, and 
amygdala. These regions are closely associated with medial temporal 
memory systems and posterior association networks implicated in 
Alzheimer-related neurodegeneration 
\citep{Dickerson2009,Bakkour2009,Verdi2021,Verdi2023}. The model was 
not developed as a diagnostic classifier and these findings should not 
be interpreted as validated disease biomarkers. After adjustment for 
age, repeated measurements, processing version, and spatial dependence, 
the proposed framework recovered coherent deviation patterns in regions 
where neurodegenerative structural variation is biologically plausible. 
At the subject level, the highest-abnormality case exhibited widespread 
multiregional deviation despite relatively preserved cognitive 
summaries, illustrating that structural deviation and global cognitive 
measures may not align at a single time point \citep{Verdi2021,
Verdi2023}.

The study has limitations. The OASIS-3 analysis used selected cortical 
and subcortical regions rather than vertex-wise or voxel-wise 
representations, limiting spatial resolution. The adjacency structure 
was anatomically specified rather than estimated from the data; future 
work could compare anatomical, connectivity-based, and data-adaptive 
spatial structures. The present formulation assumes Gaussian residual 
variation, and structural neuroimaging measurements may exhibit 
skewness or heavy tails, motivating extensions using warped or more 
flexible likelihood formulations \citep{Fraza2021}. Extension to vertex-wise settings would require sparse precision methods, reduced-rank spatial 
representations, or approximate Bayesian inference. Future work should 
also evaluate whether individualized deviation burden predicts 
longitudinal cognitive decline, biomarker progression, or treatment 
response across external cohorts.

The proposed model extends normative modeling to settings where individualized anatomical deviation evolves over time and exhibits structured dependence across brain regions. Individualized deviation is represented as a structured latent process with an explicit posterior distribution, linking normative modeling \citep{Marquand2016,RutherfordEvidence2022,Rutherford2022,Verdi2021,
Verdi2023} to longitudinal mixed-effects methodology and Bayesian spatial modeling \citep{Huertas2017,Mejia2022,Mansour2025}. The results support longitudinal spatial normative modeling for individualized neuroimaging analysis when the primary scientific interest is the regional organization of anatomical departure rather than population-average trajectories.

\section*{Disclosure statement}\label{disclosure-statement
The author declares no conflicts of interest exist.}
The author declares no conflicts of interest exist.
\section*{Data Availability Statement}\label{data-availability-statement}

The OASIS-3 dataset used in this study is publicly available through the OASIS Brains project upon approved data use agreement at \url{https://www.oasis-brains.org}. All statistical analyses were conducted in R. Code supporting the simulation study, Bayesian model fitting, posterior inference, and figure generation is available through an archived repository:
\url{https://doi.org/10.5281/zenodo.20154009}

\newpage
\clearpage

\setcounter{page}{1}
\renewcommand{\thepage}{S\arabic{page}}

\setcounter{section}{0}
\setcounter{subsection}{0}

\setcounter{figure}{0}
\renewcommand{\thefigure}{S\arabic{figure}}
\renewcommand{\theHfigure}{S\arabic{figure}}

\setcounter{table}{0}
\renewcommand{\thetable}{S\arabic{table}}
\renewcommand{\theHtable}{S\arabic{table}}

\setcounter{equation}{0}
\renewcommand{\theequation}{S\arabic{equation}}
\renewcommand{\theHequation}{S\arabic{equation}}

\renewcommand{\thesupproposition}{S\arabic{supproposition}}

\renewcommand{\thesection}{S\arabic{section}}
\renewcommand{\thesubsection}{S\arabic{section}.\arabic{subsection}}
\phantomsection\label{supplementary-material}
\bigskip

\begin{center}

{\large\bf SUPPLEMENTARY MATERIAL}\\
For\\
\textbf{A Bayesian Longitudinal Spatial Normative Model\\ For Individualized Brain Deviation Mapping}
\end{center}

% \begin{description}
% \item[A Bayesian Longitudinal Spatial Normative Model]
% Brief description. (file type)
% \item[R-package for MYNEW routine:]
% R-package MYNEW containing code to perform the diagnostic methods
% described in the article. The package also contains all datasets used as
% examples in the article. (GNU zipped tar file)
% \item[HIV data set:]
% Data set used in the illustration of MYNEW method in
% Section~\ref{sec-verify} (.txt file).
% \end{description}
 \section{Additional model details}

This supplement provides technical and computational material supporting the main manuscript. Section~S1 records additional modeling details and notation. Section~S2 presents theoretical derivations and proofs that complement the primary methodological development in the main text. Section~S3 describes the simulation implementation in greater detail, including construction of the spatial dependence structure and generation of individualized abnormality patterns. Section~S4 reports additional simulation summaries and calibration diagnostics. Section~S5 provides extended analyses for the OASIS-3 structural neuroimaging application, including preprocessing details, posterior diagnostics, supplementary cortical maps, and additional subject-level analyses. Section~S6 summarizes computational implementation details and software settings.

The proposed Bayesian longitudinal spatial normative model is
\[
Y_{itr}
=
\bm{X}_{it}^{\top}\bm{\beta}_r
+
b_i
+
u_{ir}
+
\varepsilon_{itr},
\]
where $Y_{itr}$ denotes the structural neuroimaging measurement for subject $i$, visit $t$, and region $r$, $\bm{X}_{it}$ contains visit-level covariates, $b_i$ is a subject-specific random intercept, and $u_{ir}$ represents the latent subject-specific spatial deviation process. The residual errors satisfy
\[
\varepsilon_{itr}\sim \mathcal{N}(0,\sigma^2).
\]

The subject-level random intercepts are modeled as
\[
b_i \sim \mathcal{N}(0,\sigma_b^2),
\]
and the regional spatial deviation vectors are assigned the prior
\[
\bm{u}_i
\sim
\mathcal{N}
\left(
\bm{0},
\tau_u^2\bm{Q}(\rho)^{-1}
\right),
\]
where $\bm{Q}(\rho)$ is a spatial precision matrix constructed from a region-level adjacency structure. The primary inferential target is the individualized deviation map $\bm{u}_i$, which quantifies how each subject deviates from the expected normative trajectory after accounting for age, repeated measurements, and anatomical dependence across regions.

\section{Additional theoretical results and proofs}

\subsection{Nested submodel structure}

\begin{supproposition}[Nested submodels]
Consider the hierarchical model
\[
Y_{itr}
=
\bm{X}_{it}^{\top}\bm{\beta}_r
+
b_i
+
u_{ir}
+
\varepsilon_{itr}.
\]
Several commonly used neuroimaging normative models arise as special cases under restrictions on the variance components.

\begin{enumerate}
\item If $\tau_u^2=0$, the model reduces to a longitudinal non-spatial normative model.

\item If $\sigma_b^2=0$, dependence across regions is induced only through the spatial deviation process.

\item If $\tau_u^2=0$, $\sigma_b^2=0$, and $T_i=1$ for all subjects, the model reduces to an independent cross-sectional regional model.

\item If the regional mean functions in part (3) are estimated separately using nonparametric predictive models, the resulting formulation corresponds to conventional region-wise normative modeling approaches commonly used in neuroimaging studies.
\end{enumerate}
\end{supproposition}

\begin{proof}
For part (1), setting $\tau_u^2=0$ implies
\[
u_{ir}=0
\]
almost surely for all regions. The model simplifies to
\[
Y_{itr}
=
\bm{X}_{it}^{\top}\bm{\beta}_r
+
b_i
+
\varepsilon_{itr},
\]
which retains within-subject longitudinal dependence through $b_i$ but removes spatial dependence across regions.

For part (2), setting $\sigma_b^2=0$ implies
\[
b_i=0
\]
almost surely. The model becomes
\[
Y_{itr}
=
\bm{X}_{it}^{\top}\bm{\beta}_r
+
u_{ir}
+
\varepsilon_{itr}.
\]
In this setting, the only source of dependence across regions arises from the covariance structure of the spatial deviation vector $\bm{u}_i$.

For part (3), if $\tau_u^2=0$, $\sigma_b^2=0$, and $T_i=1$, then
\[
Y_{i1r}
=
\bm{X}_{i1}^{\top}\bm{\beta}_r
+
\varepsilon_{i1r}.
\]
Since the residual terms are independent across regions, each region is modeled independently without spatial or longitudinal borrowing.

Part (4) follows by replacing the parametric regional mean structure with separate nonparametric predictive functions estimated independently for each region. This recovers the structure used in many existing region-wise neuroimaging normative models.
\end{proof}

\subsection{Posterior distribution of the spatial deviation map}

For subject $i$, define the residualized observation vector
\[
\widetilde{\bm{Y}}_{it}
=
\bm{Y}_{it}
-
\bm{B}^{\top}\bm{X}_{it}
-
b_i\bm{1}_R,
\]
where $\bm{1}_R$ denotes an $R$-dimensional vector of ones.

Conditional on $\bm{u}_i$,
\[
\widetilde{\bm{Y}}_{it}
=
\bm{u}_i
+
\bm{\varepsilon}_{it},
\]
with
\[
\bm{\varepsilon}_{it}
\sim
\mathcal{N}
(
\bm{0},
\sigma^2\bm{I}_R
).
\]

The prior density of $\bm{u}_i$ is proportional to
\[
\exp
\left\{
-
\frac{1}{2\tau_u^2}
\bm{u}_i^{\top}
\bm{Q}(\rho)
\bm{u}_i
\right\}.
\]

The likelihood contribution from all visits of subject $i$ is proportional to
\[
\prod_{t=1}^{T_i}
\exp
\left\{
-
\frac{1}{2\sigma^2}
(
\widetilde{\bm{Y}}_{it}
-
\bm{u}_i
)^{\top}
(
\widetilde{\bm{Y}}_{it}
-
\bm{u}_i
)
\right\}.
\]

Combining the prior and likelihood gives the posterior kernel
\[
\exp
\left[
-
\frac{1}{2\tau_u^2}
\bm{u}_i^{\top}
\bm{Q}(\rho)
\bm{u}_i
-
\frac{1}{2\sigma^2}
\sum_{t=1}^{T_i}
(
\widetilde{\bm{Y}}_{it}
-
\bm{u}_i
)^{\top}
(
\widetilde{\bm{Y}}_{it}
-
\bm{u}_i
)
\right].
\]

Expanding the likelihood term,
\begin{align*}
\sum_{t=1}^{T_i}
(
\widetilde{\bm{Y}}_{it}
-
\bm{u}_i
)^{\top}
(
\widetilde{\bm{Y}}_{it}
-
\bm{u}_i
)
&=
\sum_{t=1}^{T_i}
\widetilde{\bm{Y}}_{it}^{\top}
\widetilde{\bm{Y}}_{it}
-
2
\bm{u}_i^{\top}
\sum_{t=1}^{T_i}
\widetilde{\bm{Y}}_{it}
+
T_i
\bm{u}_i^{\top}\bm{u}_i.
\end{align*}

Dropping terms independent of $\bm{u}_i$ yields
\[
\exp
\left\{
-
\frac{1}{2}
\bm{u}_i^{\top}
\left(
\tau_u^{-2}\bm{Q}(\rho)
+
T_i\sigma^{-2}\bm{I}_R
\right)
\bm{u}_i
+
\bm{u}_i^{\top}
\left(
\sigma^{-2}
\sum_{t=1}^{T_i}
\widetilde{\bm{Y}}_{it}
\right)
\right\}.
\]

Completing the square gives
\[
\bm{u}_i
\mid
\{\bm{Y}_{it}\}_{t=1}^{T_i},
\bm{B},
b_i,
\sigma^2,
\tau_u^2,
\rho
\sim
\mathcal{N}
(
\bm{m}_i,
\bm{S}_i
),
\]
where
\[
\bm{S}_i
=
\left(
\tau_u^{-2}\bm{Q}(\rho)
+
T_i\sigma^{-2}\bm{I}_R
\right)^{-1},
\]
and
\[
\bm{m}_i
=
\bm{S}_i
\left(
\sigma^{-2}
\sum_{t=1}^{T_i}
\widetilde{\bm{Y}}_{it}
\right).
\]

The posterior mean $\bm{m}_i$ therefore combines subject-specific longitudinal information with anatomically structured spatial borrowing across regions.

\subsection{Posterior predictive distribution}

\begin{supproposition}[Posterior predictive distribution]
For a future visit $t^\star$ of subject $i$ with covariates $\bm{X}_{it^\star}$,
\[
\bm{Y}_{it^\star}^{\text{new}}
\mid
\{\bm{Y}_{it}\}_{t=1}^{T_i}
\sim
\mathcal{N}
\left(
\bm{B}^{\top}\bm{X}_{it^\star}
+
b_i\bm{1}_R
+
\bm{m}_i,
\,
\bm{S}_i+\sigma^2\bm{I}_R
\right).
\]
\end{supproposition}

\begin{proof}
The predictive model satisfies
\[
\bm{Y}_{it^\star}^{\text{new}}
=
\bm{B}^{\top}\bm{X}_{it^\star}
+
b_i\bm{1}_R
+
\bm{u}_i
+
\bm{\varepsilon}_{it^\star}^{\text{new}},
\]
where
\[
\bm{\varepsilon}_{it^\star}^{\text{new}}
\sim
\mathcal{N}
(
\bm{0},
\sigma^2\bm{I}_R
)
\]
and is conditionally independent of $\bm{u}_i$.

Since
\[
\bm{u}_i
\mid
\cdots
\sim
\mathcal{N}
(
\bm{m}_i,
\bm{S}_i
),
\]
the sum
\[
\bm{u}_i
+
\bm{\varepsilon}_{it^\star}^{\text{new}}
\]
is Gaussian with mean $\bm{m}_i$ and covariance
\[
\bm{S}_i+\sigma^2\bm{I}_R.
\]

Adding the deterministic mean component gives the stated predictive distribution.
\end{proof}

\subsection{Oracle standardization}

\begin{supproposition}[Oracle standardization]
Assume the normative model is correctly specified and that all model parameters are known. Let
\[
\mu_{itr}
=
\E(Y_{itr}\mid \text{history}),
\]
and
\[
v_{itr}
=
\Var(Y_{itr}\mid \text{history}).
\]
Then
\[
Z_{itr}^{\circ}
=
\frac{
Y_{itr}
-
\mu_{itr}
}{
\sqrt{v_{itr}}
}
\sim
\mathcal{N}(0,1).
\]
\end{supproposition}

\begin{proof}
Under correct model specification,
\[
Y_{itr}
\mid
\text{history}
\sim
\mathcal{N}
(
\mu_{itr},
v_{itr}
).
\]

Subtracting the true predictive mean and dividing by the true predictive standard deviation standardizes the Gaussian variable. Therefore,
\[
\frac{
Y_{itr}
-
\mu_{itr}
}{
\sqrt{v_{itr}}
}
\sim
\mathcal{N}(0,1).
\]
\end{proof}

\subsection{Bayes optimality of spatial borrowing}

A principal motivation for introducing $\bu_i$ with structured covariance is that spatial borrowing should improve estimation of individualized deviation maps when the regional dependence model is correctly specified. This can be stated formally.

\begin{supproposition}[Posterior mean as Bayes-optimal deviation estimator]
\label{prop:s_bayes_optimal}
Fix subject $i$ and suppose the hierarchical model is correctly specified. Under squared error loss
\[
L(\widehat{\bu}_i,\bu_i)=\|\widehat{\bu}_i-\bu_i\|_2^2,
\]
the posterior mean
\[
\widehat{\bu}_i^{\mathrm{Bayes}}
=
\E\!\left(
\bu_i
\mid
\{\bm{Y}_{it}\}_{t=1}^{T_i},
\bm{B},
b_i,
\sigma^2,
\tau_u^2,
\rho
\right)
=
\bm{m}_i
\]
minimizes posterior expected loss over all measurable estimators of $\bu_i$ based on subject $i$'s data.
\end{supproposition}

\begin{proof}
Let $\widehat{\bu}_i$ be any measurable estimator based on the observed data for subject $i$. Conditional on the observed data and model parameters, the posterior expected loss is
\[
R(\widehat{\bu}_i)
=
\E\!\left[
\|\widehat{\bu}_i-\bu_i\|_2^2
\mid
\{\bm{Y}_{it}\}_{t=1}^{T_i},
\bm{B},
b_i,
\sigma^2,
\tau_u^2,
\rho
\right].
\]
Let $\bm{m}_i=\E(\bu_i\mid\{\bm{Y}_{it}\}_{t=1}^{T_i},\bm{B},b_i,\sigma^2,\tau_u^2,\rho)$. Expanding the squared loss around $\bm{m}_i$ gives
\begin{align*}
\|\widehat{\bu}_i-\bu_i\|_2^2
&=
\|\widehat{\bu}_i-\bm{m}_i+\bm{m}_i-\bu_i\|_2^2 \\
&=
\|\widehat{\bu}_i-\bm{m}_i\|_2^2
+
\|\bm{m}_i-\bu_i\|_2^2
+
2(\widehat{\bu}_i-\bm{m}_i)^\top(\bm{m}_i-\bu_i).
\end{align*}
Taking posterior expectation conditional on the observed data yields
\begin{align*}
R(\widehat{\bu}_i)
&=
\|\widehat{\bu}_i-\bm{m}_i\|_2^2
+
\E\!\left[
\|\bm{m}_i-\bu_i\|_2^2
\mid
\{\bm{Y}_{it}\}_{t=1}^{T_i},
\bm{B},
b_i,
\sigma^2,
\tau_u^2,
\rho
\right] \\
&\quad+
2(\widehat{\bu}_i-\bm{m}_i)^\top
\E\!\left[
\bm{m}_i-\bu_i
\mid
\{\bm{Y}_{it}\}_{t=1}^{T_i},
\bm{B},
b_i,
\sigma^2,
\tau_u^2,
\rho
\right].
\end{align*}
By definition of $\bm{m}_i$ as the posterior mean,
\[
\E\!\left[
\bm{m}_i-\bu_i
\mid
\{\bm{Y}_{it}\}_{t=1}^{T_i},
\bm{B},
b_i,
\sigma^2,
\tau_u^2,
\rho
\right]
=
\bm{0}.
\]
Therefore,
\[
R(\widehat{\bu}_i)
=
\|\widehat{\bu}_i-\bm{m}_i\|_2^2
+
\E\!\left[
\|\bm{m}_i-\bu_i\|_2^2
\mid
\{\bm{Y}_{it}\}_{t=1}^{T_i},
\bm{B},
b_i,
\sigma^2,
\tau_u^2,
\rho
\right].
\]
The second term does not depend on $\widehat{\bu}_i$, and the first term is minimized uniquely at $\widehat{\bu}_i=\bm{m}_i$. Hence the posterior mean $\bm{m}_i$ is the Bayes estimator of $\bu_i$ under squared error loss.
\end{proof}

\begin{remark}
This result gives a decision-theoretic interpretation to the proposed spatial deviation map. When the hierarchical model and the spatial covariance structure are correctly specified, the posterior mean $\bm{m}_i$ is not merely a smoothed residual vector. It is the Bayes-optimal estimator of the individualized spatial deviation process under squared error loss. Since $\bm{m}_i$ depends on $\bQ(\rho)$ through the posterior covariance and precision structure, the spatial adjacency model directly enters the optimal estimator.
\end{remark}

\section{Additional simulation implementation details}

\subsection{Simulation design}

The simulation study was designed to evaluate recovery of individualized deviation maps under varying spatial and longitudinal data structures. Each simulated dataset contained repeated measurements from multiple subjects observed across multiple brain regions. Data were generated directly from the hierarchical structure assumed by the proposed Bayesian longitudinal spatial normative model.

The simulations varied several characteristics relevant to longitudinal neuroimaging studies, including the strength of spatial dependence, nonlinear mean trajectories, irregular visit schedules, and missing follow-up. Six scenarios were considered in total:
\begin{enumerate}
\item baseline linear longitudinal structure,
\item moderate spatial dependence,
\item strong spatial dependence,
\item variable visit schedules,
\item missing longitudinal follow-up,
\item nonlinear mean trajectories.
\end{enumerate}

For each simulated dataset, individualized abnormality maps were generated through latent subject-specific spatial deviation vectors. The same covariate structure and region-level adjacency assumptions were then used during model fitting.

\subsection{Construction of the spatial adjacency matrix}

The spatial precision matrix was constructed from a region-level adjacency matrix
\[
\bm{W}.
\]
The matrix encoded neighborhood relationships between anatomically related regions. The diagonal degree matrix was defined as
\[
\bm{D}_{rr}
=
\sum_{r'}
W_{rr'}.
\]

The precision matrix used throughout the simulations was
\[
\bm{Q}(\rho)
=
\bm{D}
-
\rho\bm{W},
\]
where $\rho$ controlled the strength of spatial dependence.

To ensure numerical stability and positive definiteness during posterior computation, a small diagonal ridge term was added when necessary. Different simulation scenarios used different values of $\rho$ to represent weak, moderate, and strong anatomical dependence.

\subsection{Simulation performance metrics}

Model performance was evaluated using both estimation accuracy and calibration criteria. Accuracy metrics included mean bias, mean squared error, and deviation-map mean squared error. The deviation-map mean squared error was used as the primary measure of individualized abnormality recovery because the principal inferential target of the proposed framework is the latent subject-specific deviation map.

Calibration metrics were based on the standardized deviation scores produced by each fitted model. Specifically, the simulations evaluated:
\[
\E(Z),
\qquad
\Var(Z),
\qquad
\Pr(|Z|>1.96).
\]

Well-calibrated normative models should produce approximately standard normal deviation scores under correct specification. Deviations from these targets indicate underestimation or overestimation of predictive uncertainty.

\section{Additional simulation results}

The primary simulation summaries are reported in Tables~\ref{tab:accuracy}--\ref{tab:calibration} and Figures~\ref{fig:accuracy}--\ref{fig:calibration} of the main manuscript. Additional simulation diagnostics examined during model validation did not materially alter the reported findings and are therefore not presented separately.

Across simulation replicates, Bayesian model fitting remained stable, with representative fitted models exhibiting satisfactory convergence behavior, no divergent transitions, and $\hat{R}$ values close to 1 across monitored parameters. Calibration summaries further showed that standardized deviation scores remained close to their target reference behavior across the simulated settings.
\section{Extended real-data analysis}

\subsection{OASIS-3 structural neuroimaging data}

The real-data application used longitudinal structural neuroimaging measurements from the OASIS-3 study \citep{LaMontagne2019}. OASIS-3 is a longitudinal neuroimaging, clinical, cognitive, and biomarker dataset designed to support studies of healthy aging and Alzheimer-type neurodegeneration. The dataset includes repeated MRI observations collected across multiple visits together with demographic and clinical information.

The present analysis focused on FreeSurfer-derived cortical thickness and subcortical volume measures. Structural measurements were merged with visit-level demographic and clinical variables and converted into long format so that each row represented a subject-region-visit combination.

Subjects were retained if they had at least two MRI visits after preprocessing and quality-control procedures. Observations with missing regional measurements, missing age information, or incomplete covariate data were excluded prior to model fitting.

\subsection{Regional measures and covariates}

The analysis focused on cortical and subcortical regions commonly implicated in aging and neurodegenerative disease progression. These included bilateral hippocampus, amygdala, entorhinal cortex, parahippocampal cortex, temporal pole, inferior temporal cortex, middle temporal cortex, posterior cingulate cortex, and precuneus.

Regional measures were standardized within region before modeling:
\[
Y_{itr}^{\text{std}}
=
\frac{
Y_{itr}
-
\bar{Y}_r
}{
s_r
},
\]
where $\bar{Y}_r$ and $s_r$ denote the sample mean and sample standard deviation for region $r$. Standardization allowed cortical thickness and subcortical volume measurements to be analyzed jointly on a comparable scale.

Age at MRI visit was standardized and included as the primary longitudinal covariate. FreeSurfer processing version was additionally included as a categorical adjustment variable to account for possible processing-version related variation in structural measurements.

\subsection{Real-data model specification}

Three models were compared in the OASIS-3 application:
\begin{enumerate}
\item an independent cross-sectional regional model,
\item a longitudinal non-spatial mixed-effects model,
\item the proposed Bayesian subject-specific spatial model.
\end{enumerate}

The proposed Bayesian model was
\[
Y_{itr}
=
\alpha_r
+
\beta_rA_{it}
+
\gamma_{v[it]}
+
b_i
+
u_{ir}
+
\varepsilon_{itr},
\]
where $A_{it}$ denotes standardized age at visit and $\gamma_{v[it]}$ denotes the FreeSurfer-version effect.

The spatial adjacency structure connected anatomically related regions and bilateral homologous regions. Medial temporal, lateral temporal, and posterior cortical regions were additionally grouped into broader anatomical families to induce biologically plausible spatial borrowing across neighboring structures.

\subsection{Posterior computation and convergence assessment}

The Bayesian model was implemented in Stan using the \texttt{cmdstanr} interface \citep{carpenter2017stan}. Four Markov chains were run with 1,000 warmup iterations and 1,000 post-warmup sampling iterations per chain. The target acceptance probability was set to 0.99 and the maximum tree depth was set to 13.

Convergence was evaluated using traceplots, the potential scale reduction statistic $\widehat{R}$, and effective sample size summaries. Posterior summaries for the key variance parameters were:
\[
\sigma
=
0.491,
\qquad
\sigma_b
=
0.517,
\qquad
\tau_u
=
0.928,
\qquad
\sigma_{\text{version}}
=
0.305.
\]

All key parameters exhibited satisfactory convergence behavior with $\widehat{R}$ values close to one and adequate effective sample sizes.

\subsection{Additional real-data findings}

The Bayesian subject-specific spatial model produced substantially improved individualized fit relative to the benchmark models. The residual standard deviation decreased from 0.926 under the independent cross-sectional model and 0.774 under the longitudinal non-spatial model to 0.423 under the Bayesian spatial model. The proposed framework also produced lower mean absolute deviation and improved calibration of standardized deviation scores.

Regional summaries identified the temporal pole, entorhinal cortex, inferior temporal cortex, posterior cingulate cortex, and amygdala as regions exhibiting the greatest burden of extreme deviations. These findings are biologically consistent with regions commonly implicated in aging and Alzheimer-type neurodegenerative processes.

Subject-level analyses further demonstrated substantial heterogeneity in individualized abnormality burden across participants. Several subjects exhibited persistent extreme deviation trajectories across repeated visits despite relatively preserved cognitive scores, illustrating the ability of the proposed framework to identify individualized structural abnormality patterns beyond simple group-average summaries.

\subsection{Supplementary real-data figures}
\begin{figure}[!ht]
\centering
\includegraphics[width=3in,height=\textheight]{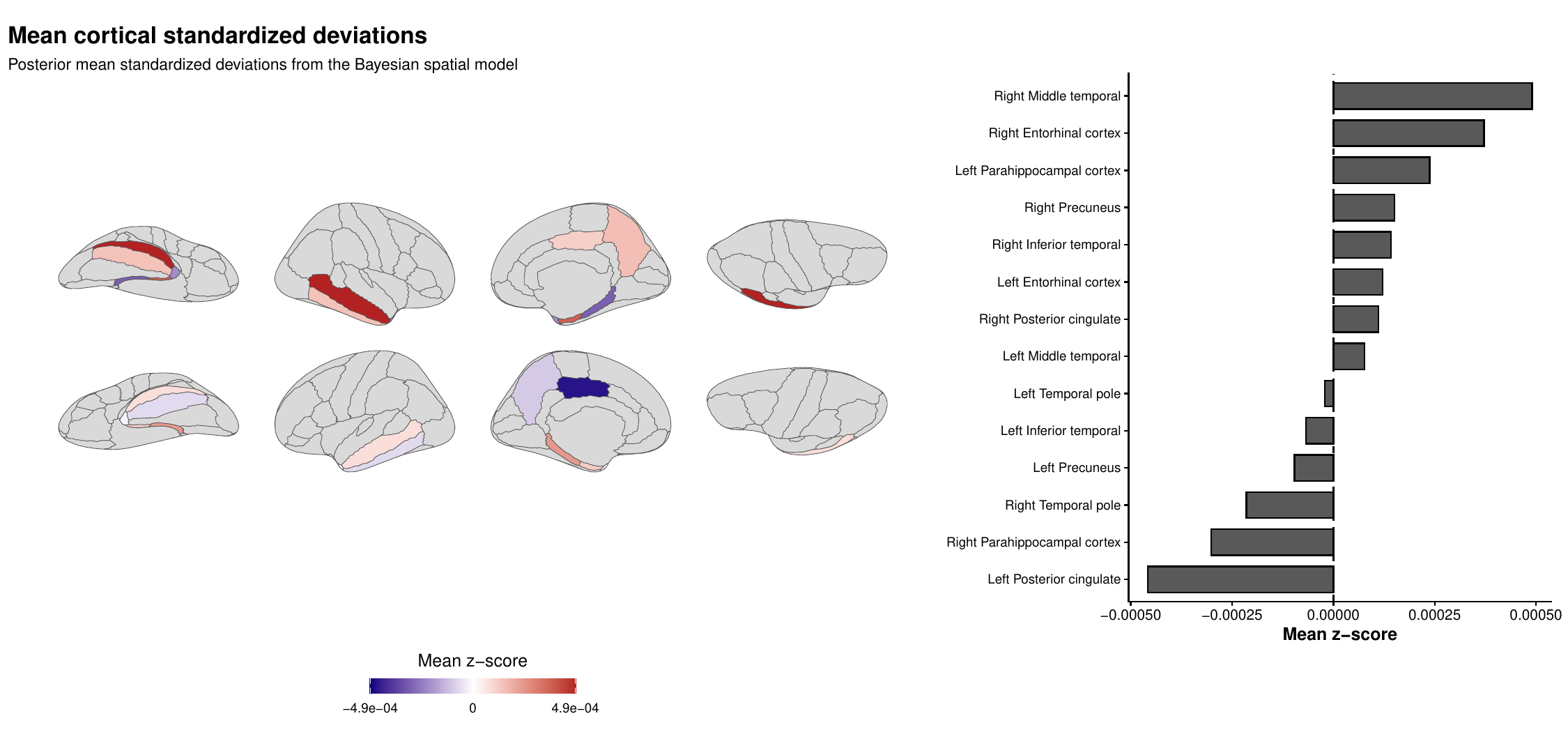}
\caption{Mean cortical standardized deviations from the Bayesian subject-specific spatial model. This figure complements the main cortical extreme-deviation map by displaying posterior mean deviations rather than tail probabilities.}
\label{fig:s_mean_cortical}
\end{figure}

\begin{figure}[!ht]
\centering
\includegraphics[width=3.5in,height=\textheight]{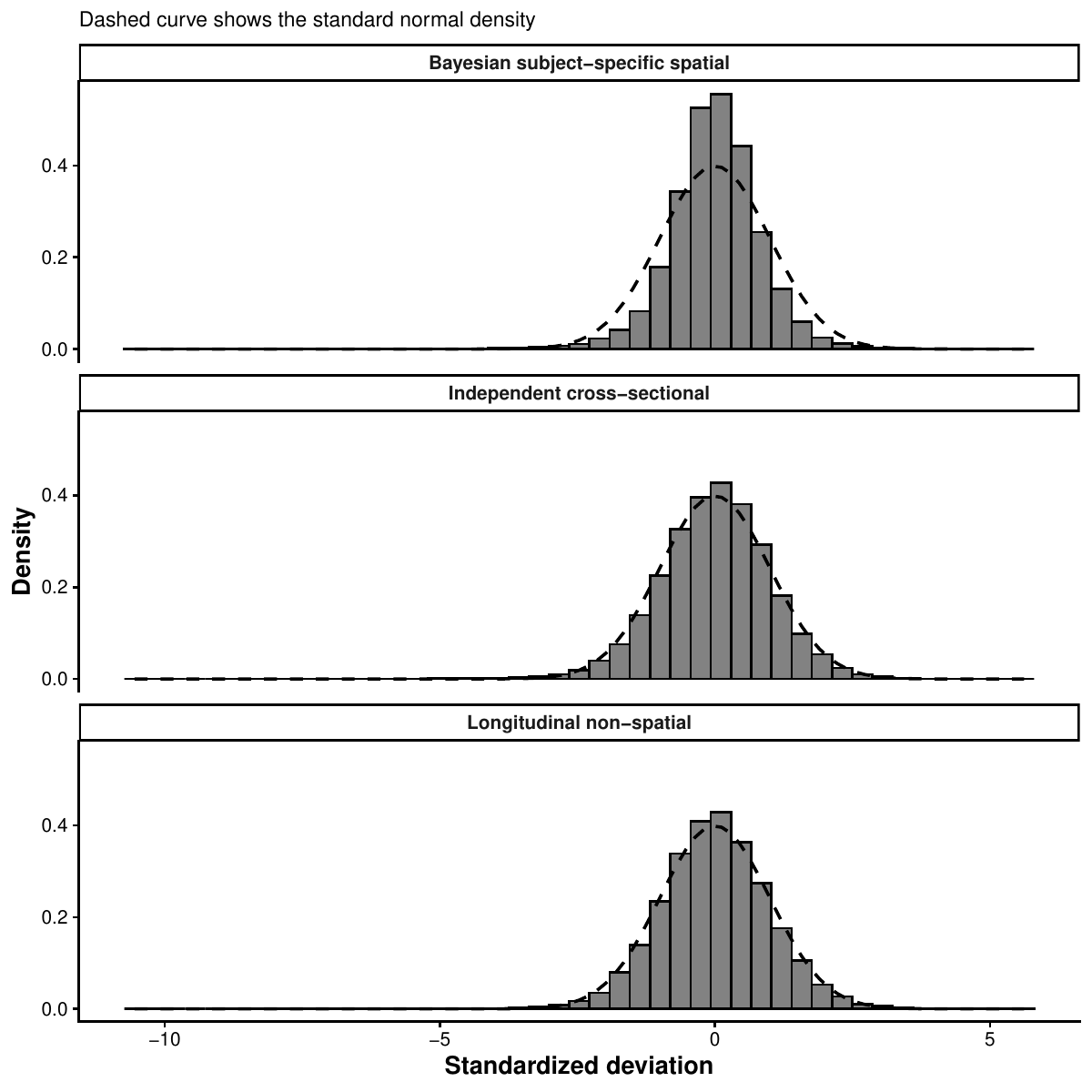}
\caption{Calibration histograms for standardized deviation scores across the three OASIS-3 models. The dashed curve represents the standard normal density.}
\label{fig:s_histograms}
\end{figure}

\begin{figure}[!ht]
\centering
\includegraphics[width=3in,height=\textheight]{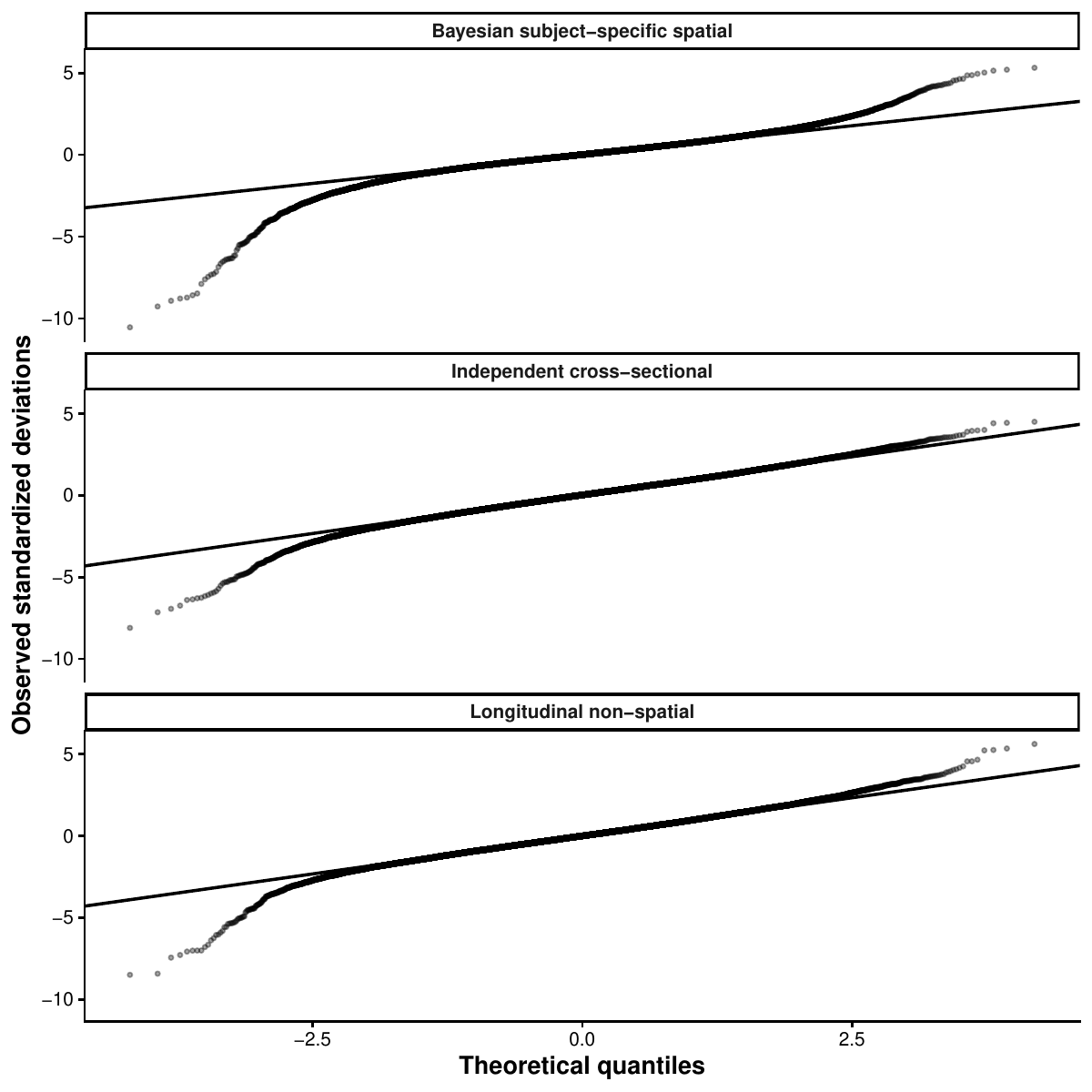}
\caption{Normal QQ plots for standardized deviation scores across the independent cross-sectional, longitudinal non-spatial, and Bayesian subject-specific spatial models.}
\label{fig:s_qq}
\end{figure}

\begin{figure}[!ht]
\centering
\includegraphics[width=3in,height=\textheight]{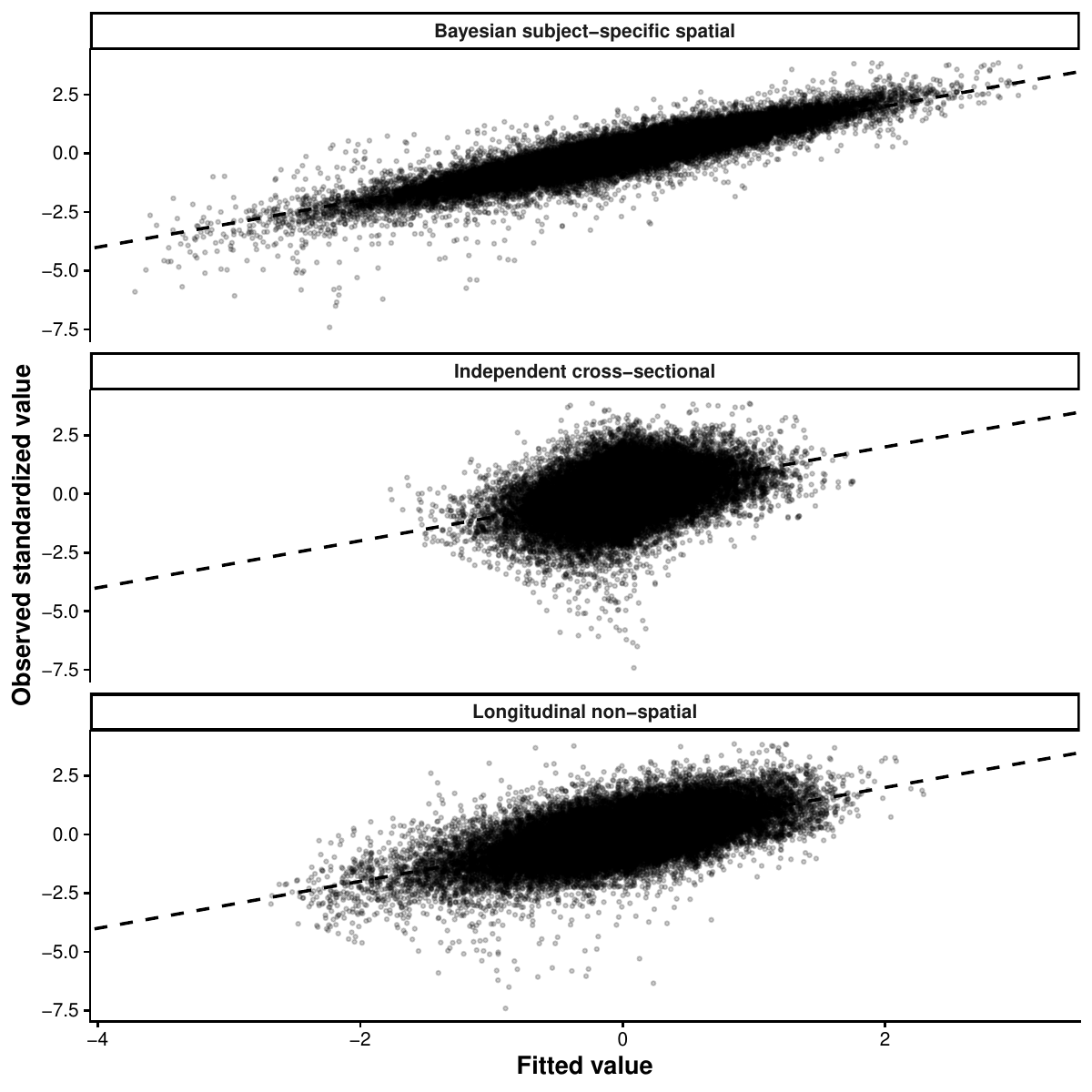}
\caption{Observed versus fitted standardized structural measurements across the three OASIS-3 models. Points closer to the diagonal indicate better agreement between fitted and observed standardized values.}
\label{fig:s_obs_pred}
\end{figure}

\begin{figure}[!ht]
\centering
\includegraphics[width=3in,height=\textheight]{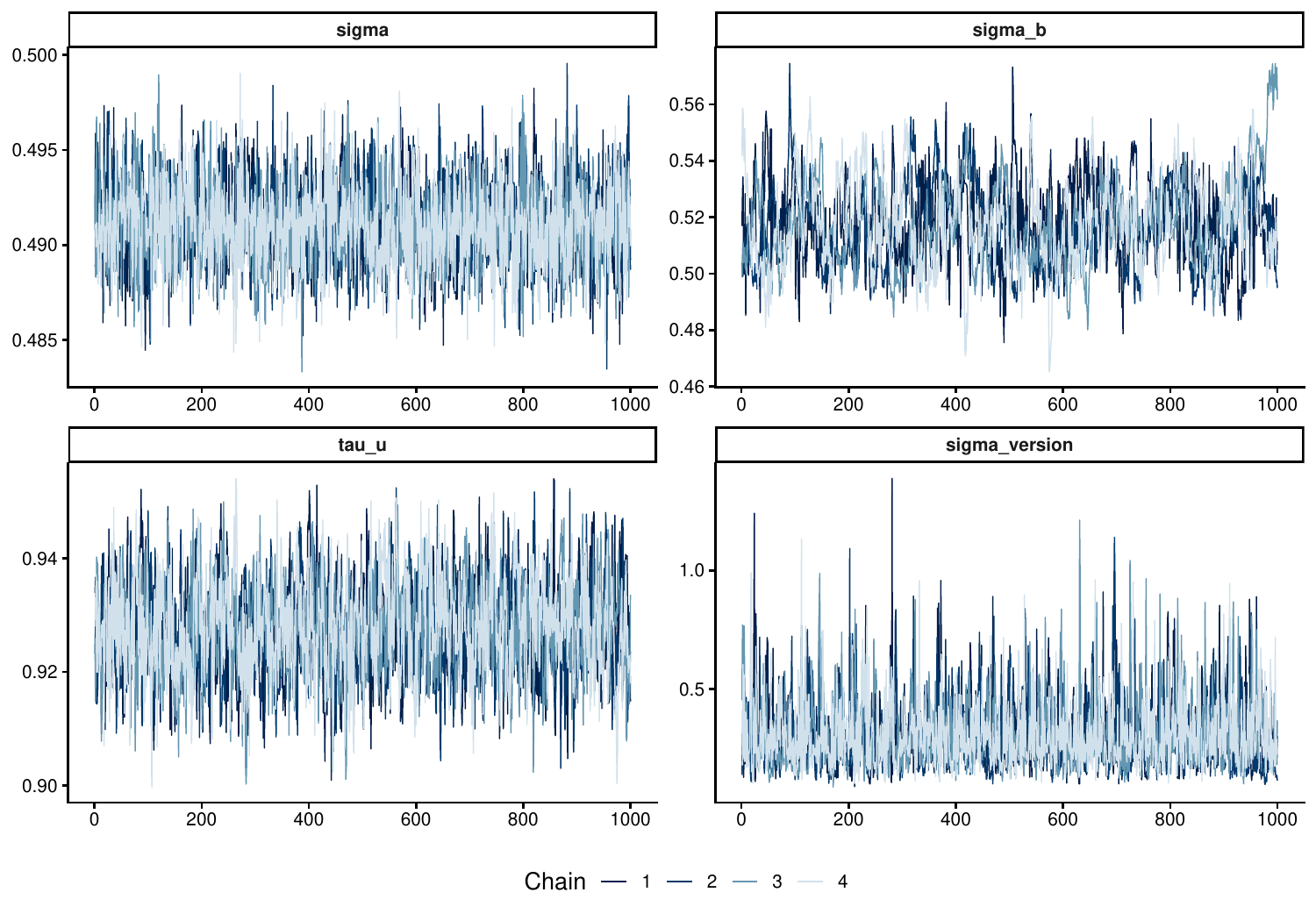}
\caption{MCMC traceplots for key posterior scale parameters in the Bayesian subject-specific spatial model. Chains showed stable mixing around common posterior ranges.}
\label{fig:s_traceplots}
\end{figure}

\begin{figure}[!ht]
\centering
\includegraphics[width=3in,height=\textheight]{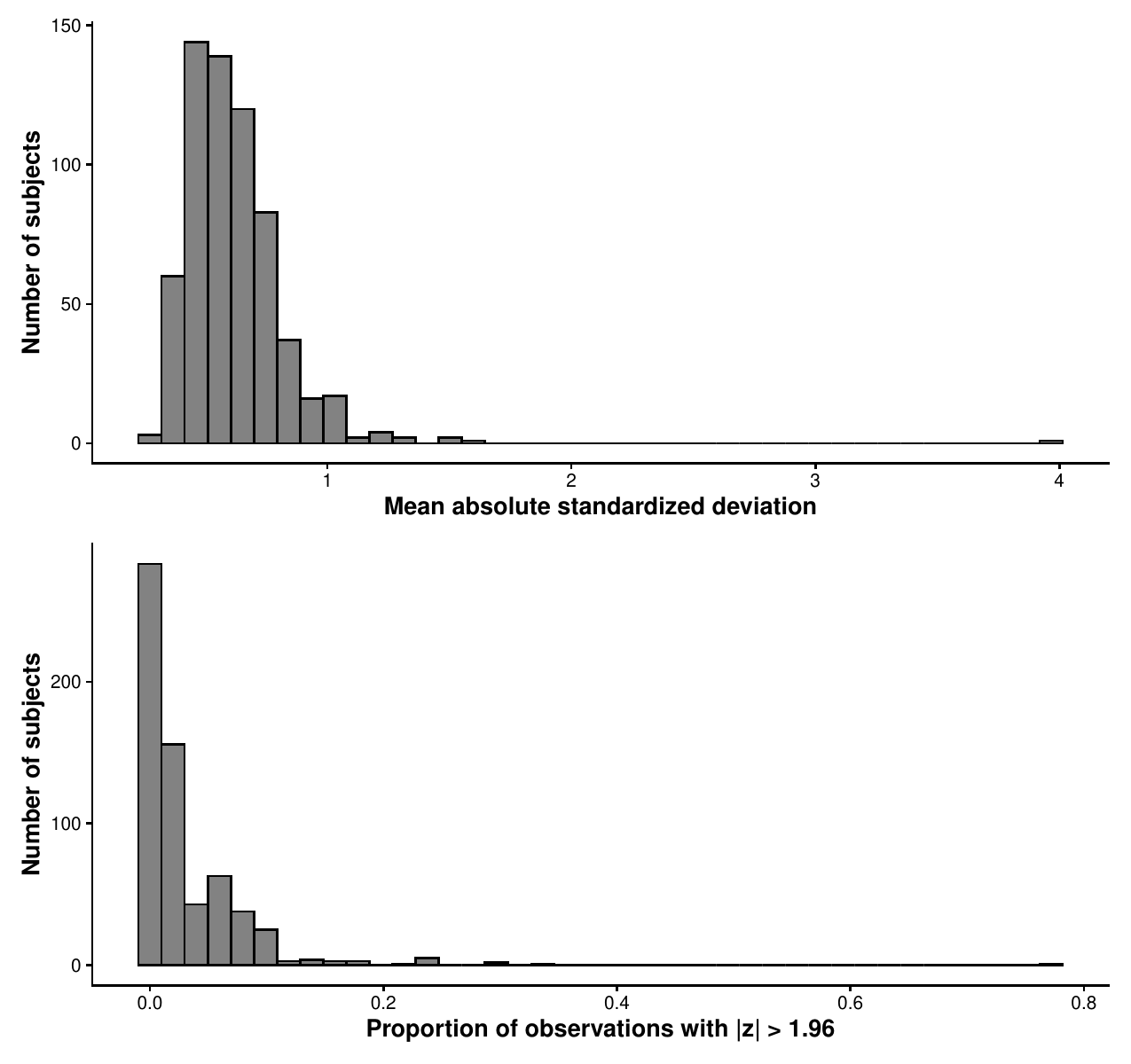}
\caption{Distribution of subject-level abnormality burden in the OASIS-3 application, summarized by mean absolute standardized deviation and the proportion of observations exceeding $|Z|>1.96$.}
\label{fig:s_subject_burden}
\end{figure}

\newpage
\subsection{Sensitivity analyses}
Multiple sensitivity analyses were considered to evaluate the robustness of the proposed framework. These included alternative specifications of the spatial adjacency matrix, exclusion of subjects with unusually large deviation burdens, and comparison with models excluding spatial dependence.
\begin{longtable}[]{@{}lrrrrrrr@{}}
\caption{Posterior summaries for key variance and scale parameters in the OASIS-3 Bayesian subject-specific spatial model.}
\label{tab:s_post_summary}\tabularnewline

\toprule\noalign{}
Parameter & Mean & Median & SD & 5\% & 95\% & $\widehat{R}$ & Bulk ESS \\
\midrule\noalign{}
\endfirsthead

\toprule\noalign{}
Parameter & Mean & Median & SD & 5\% & 95\% & $\widehat{R}$ & Bulk ESS \\
\midrule\noalign{}
\endhead

\bottomrule\noalign{}
\endlastfoot

$\sigma$
& 0.4911 & 0.4911 & 0.0023 & 0.4875 & 0.4949 & 1.0024 & 2741.8 \\

$\sigma_b$
& 0.5174 & 0.5164 & 0.0148 & 0.4952 & 0.5425 & 1.0077 & 226.0 \\

$\tau_u$
& 0.9277 & 0.9277 & 0.0085 & 0.9140 & 0.9416 & 1.0025 & 1262.2 \\

$\sigma_{\mathrm{version}}$
& 0.3047 & 0.2669 & 0.1472 & 0.1405 & 0.5874 & 0.9995 & 1612.3 \\

\end{longtable}

\begin{longtable}[]{@{}lrrrrrr@{}}
\caption{Selected subjects used for the illustrative low-deviation and high-abnormality case-study profiles.}
\label{tab:s_case_subjects}\tabularnewline

\toprule\noalign{}
Subject &
Observations &
Visits &
Mean $|Z|$ &
Max $|Z|$ &
Extreme proportion &
Mean MMSE \\
\midrule\noalign{}
\endfirsthead

\toprule\noalign{}
Subject &
Observations &
Visits &
Mean $|Z|$ &
Max $|Z|$ &
Extreme proportion &
Mean MMSE \\
\midrule\noalign{}
\endhead

\bottomrule\noalign{}
\endlastfoot

OAS30765
& 57 & 3 & 0.3481 & 1.2055 & 0.0000 & 26.33 \\

OAS30303
& 57 & 3 & 3.9663 & 10.5478 & 0.7719 & 29.67 \\

\end{longtable}
The principal findings remained stable across these analyses. In particular, the Bayesian spatial model consistently produced lower deviation-map reconstruction error and more stable individualized abnormality estimates than the benchmark models. Regional abnormality patterns remained concentrated within temporolimbic and posterior cortical regions across alternative specifications.
\begin{longtable}[]{@{}p{6.5cm}rrrr@{}}
\caption{Regions with the largest extreme-deviation burden in the OASIS-3 application.}
\label{tab:s_top_regions}\tabularnewline

\toprule\noalign{}
Region &
Observations &
Mean $Z$ &
SD $Z$ &
$\Pr(|Z|>1.96)$ \\
\midrule\noalign{}
\endfirsthead

\toprule\noalign{}
Region &
Observations &
Mean $Z$ &
SD $Z$ &
$\Pr(|Z|>1.96)$ \\
\midrule\noalign{}
\endhead

\bottomrule\noalign{}
\endlastfoot

Right temporal pole thickness
& 1902 & -0.0002 & 1.1064 & 0.0778 \\

Left temporal pole thickness
& 1902 & 0.0000 & 1.0881 & 0.0705 \\

Right entorhinal thickness
& 1902 & 0.0004 & 1.0196 & 0.0547 \\

Left entorhinal thickness
& 1902 & 0.0001 & 0.9713 & 0.0515 \\

Right inferior temporal thickness
& 1902 & 0.0001 & 0.9305 & 0.0410 \\

Left inferior temporal thickness
& 1902 & -0.0001 & 0.9271 & 0.0379 \\

Left posterior cingulate thickness
& 1902 & -0.0005 & 1.0108 & 0.0352 \\

Left amygdala volume
& 1902 & 0.0003 & 0.8713 & 0.0342 \\

Right amygdala volume
& 1902 & -0.0002 & 0.9169 & 0.0331 \\

Left middle temporal thickness
& 1902 & 0.0001 & 0.8350 & 0.0289 \\

\end{longtable}
\section{Computational details}

All analyses were conducted in R using \texttt{cmdstanr}, \texttt{posterior}, \texttt{ggplot2}, \texttt{dplyr}, \texttt{Matrix}, and \texttt{brms}. Bayesian posterior inference was performed using Hamiltonian Monte Carlo implemented in Stan.

The primary computational settings included:
\begin{itemize}
\item four Markov chains,
\item 1{,}000 warmup iterations per chain,
\item 1{,}000 post-warmup sampling iterations per chain,
\item target acceptance probability of 0.99,
\item maximum tree depth of 13.
\end{itemize}

Posterior convergence was assessed using traceplots, effective sample size summaries, and $\widehat{R}$ statistics. Posterior summaries were computed using the \texttt{posterior} package.

The complete analysis pipeline, including preprocessing, simulation generation, Bayesian model fitting, and visualization, was implemented using reproducible scripts.

\bibliography{bibliography.bib}

\end{document}